\documentclass{article}

\usepackage[utf8]{inputenc}
\usepackage[english]{babel}
\usepackage{xcolor,colortbl}
\usepackage{url}
\usepackage{graphicx}
\usepackage{amsfonts}
\usepackage{amsmath,amssymb}
\usepackage{amsthm}
\usepackage[linesnumbered, ruled, vlined]{algorithm2e}
\usepackage{balance}
\usepackage{subcaption}
\usepackage{multirow}

\usepackage{typed-checklist}

\usepackage{enumitem}

\usepackage{wrapfig}

\DeclareGraphicsExtensions{.png,.jpg,.pdf}
\graphicspath{{../body/figures/}}

\newcommand{\silence}[1]{}

\definecolor{orange}{rgb}{0.93, 0.57, 0.13}

\def\other1#1{\emph{\color{magenta}O1: #1}}
\def\other2#1{\emph{\color{red}O2: #1}}
\def\other3#1{\emph{\color{purple}O3: #1}}

\newcommand{\result}[1]{\textcolor{blue}{\emph{#1}}}

\def\ANALYZE{\textsf{ANALYZE} }

\newif\ifDRAFT
\DRAFTfalse       

\ifDRAFT
\newcommand{\isDraft}[1]{#1}
\else
\newcommand{\isDraft}[1]{}
\fi

\ifDRAFT
\def\inlineRem#1{{\noindent\bf\fcolorbox{red}{white}{\color{blue}{$\|$~{#1}~$\|$}}}}
\else
\newcommand{\inlineRem}[1]{}
\fi

\ifDRAFT
\def\rem#1{\colorbox{yellow}{\scriptsize\color{brown} @@} \marginpar{\scriptsize\color{brown} #1}}
\else
\newcommand{\rem}[1]{}
\fi

\ifDRAFT
\def\sideNote#1{\marginpar{\indent \color{gray}{\footnotesize{\emph{#1}}}}}
\else
\newcommand{\sideNote}[1]{}
\fi

\def\sticky#1{\colorbox{yellow}{\scriptsize\color{blue} @@} \marginpar{\scriptsize\color{blue} #1}}

\newcommand{\forcepagebreak}{\clearpage \newpage}

\theoremstyle{definition}

\newtheorem{theorem}{Theorem}[section]

\title{Semantics and Multi-Query Optimization Algorithms for the Analyze Operator}

\author{
    Marios Iakovidis  $^{[0009-0006-9061-3450]}$ \\
	University of Ioannina, Ioannina 45110, Greece \\
    \textsf{miakovidis@uoi.gr}\\
    \and
    Panos Vassiliadis $^{[0000-0003-0085-6776]}$  \\
	University of Ioannina, Ioannina 45110, Greece \\
	\textsf{pvassil@cs.uoi.gr}
}

\begin{document}

    \maketitle

\fbox{%
  \parbox{0.9\textwidth}{%
    \footnotesize {A short version of this paper has been accepted at DOLAP 2026:\\ M. Iakovidis, P. Vassiliadis. Multi-Query Optimization for the novel Analyze Operator. Accepted at 28th International Workshop on Design, Optimization, Languages and Analytical Processing of Big Data (DOLAP) co-located with EDBT/ICDT 2026, Tampere, Finland - March 24, 2026.
    }
  }
}

\begin{abstract}
In their hunt for highlights, i.e., interesting patterns in the data, data analysts have to issue \textit{groups} of related queries and \textit{manually} combine their results. To the extent that the analyst goals are based on an \textit{intention} on what to discover (e.g., contrast a query result to peer ones, verify a pattern to a broader range of data in the data space, etc), the integration of \textit{intentional} query operators 
in analytical engines can enhance the efficiency of these analytical tasks. In this paper, we introduce, with well-defined semantics, the \textit{ANALYZE operator}, a novel cube querying intentional operator that provides a $360^{o}$ view of data. We define the semantics of an ANALYZE query as a tuple of five internal, facilitator cube queries, that (a) report on the specifics of a particular subset of the data space, which is part of the query specification, and to which we refer as the \textit{original query}, (b) contrast the result with results from peer-subspaces, or \textit{sibling queries}, and, (c) explore the data space in lower levels of granularity via 
\textit{drill-down queries}. 
We introduce formal query semantics for the operator and we theoretically prove that we can obtain the exact same result by merging the facilitator cube queries into a smaller number of queries. This effectively introduces a \textit{multi-query optimization (MQO)} strategy for executing an ANALYZE query. We propose three alternative algorithms, (a) a simple execution without optimizations (\textit{Min-MQO}), (b) a total merging of all the facilitator queries to a single one (\textit{Max-MQO}), and (c) an intermediate strategy, \textit{Mid-MQO}, that merges only a subset of the facilitator queries. Our experimentation demonstrates that Mid-MQO achieves consistently strong performance across several contexts, Min-MQO always follows it, and Max-MQO excels for queries where the siblings are sizable and significantly overlap.



\end{abstract}


\section{Introduction}\label{sec:intro}
Routinely, data analysts find themselves issuing \textit{groups} of related queries in an attempt to uncover hidden gems in the data. Such significant findings, or \textit{highlights}, are practically verifications of the existence of interesting properties  (e.g., a trend, a peak, an outlier, etc) in a subset of the data that is of interest. Automating the extraction of highlights is, therefore, an important task in accelerating the work of data analysts, and allowing for more extensive and complete exploration of the data space. 
To the extent that the underlying analyst goals are based on an \textit{intention} on what to discover, (e.g., contrast a query result to results of queries exploring subsets of the data space that are similar to the one under scrutiny, verify a pattern to a broader range of data in the dataspace, forecast the trend etc) we have introduced the need for \textit{intentional} query operators \cite{DBLP:conf/dolap/VassiliadisM18, DBLP:journals/is/VassiliadisMR19} that abstract these data explorations into higher level operators. 
As such, these operators should come with a formal model for data, formal semantics, and, of course,  with optimization methods for their execution.

Related work on the subject of \textit{Automated Highlight Extraction} provides tools for multidimensional data exploration to discover interesting data patterns \cite{DBLP:conf/vldb/Sarawagi99, DBLP:conf/edbt/SarawagiAM98,DBLP:journals/tvcg/WangSZCXMZ20, DBLP:journals/vldb/AbuzaidKSGXSASM21, DBLP:conf/sigmod/00040HZ21}. The primary focus of these tools is to explore a data space, find data patterns, and measure how interesting they are using interestingness metrics. However, the research landscape is characterized by (a) the absence of a widely accepted formal model that encompasses the proposed tasks in a coherent framework (the Intentional Model of \cite{DBLP:conf/dolap/VassiliadisM18, DBLP:journals/is/VassiliadisMR19} is a first attempt), (b) at large, the omission of the opportunity to utilize multidimensional hierarchical data spaces (in the Business Intelligence sense), by focusing almost solely on relational data, and (c) to the best of our knowledge, the absence of an operator that provides an all-encompassing view of the data of interest to an analyst in a single invocation. 

In this paper, we propose the \emph{ANALYZE operator}, a novel cube query operator that formalizes the intention of providing a $360^{o}$ view of the data to the data analyst, via a single invocation. Being an algebraic operator, ANALYZE comes with (a) a well-defined multidimensional, hierarchical underlying data model with data cubes and dimensions as its basic elements \cite{PV22}, (b) well defined semantics on the expected results, and also, in our case, (c) optimization strategies for its efficient execution. The simplicity of a hierarchical multidimensional model facilitates a simple query operator that exploits the model's ability to compute sibling, ancestor and descendant values at different levels of coarseness smoothly. 
To facilitate comprehension, we will use the following query, as a reference example throughout the paper:

\hrulefill

    \begin{tabular}{ll}
        \textsc{analyze} &{\color{blue}$sum(Store~Sales)$ as $SumSales$}\\ 
        \textsc{from} &{\color{blue}$Sales$}\\ 
        \textsc{for} &
            {\color{blue}$Date.Quarter~=~1997-Q3$ $\wedge$} \\
            &{\color{blue}$Customer.state~=~'CA'$ $\wedge$} \\
            &{\color{blue}$Promo.Media~=~'Daily~Paper'$} 
         \\
        \textsc{group by} &{\color{blue}$month,~customerRegion$}\\
        \textsc{as} &{\color{blue}$PaperPromoCA1997Q3$} \\
    \end{tabular}

\hrulefill

The semantics of an ANALYZE query include five internal, \textit{facilitator} cube queries. First, the operator has to deal with the information requested for a very specific subset of the data space that is obtained via (a) a set of filters and (b) a pair of groupers that are used in a GROUP-BY fashion to aggregate a measure. We call this the \textit{original query} of the operator as it targets a subset of the data space of interest. 
In our example here, the three filters on 1997-Q3, CA and Daily Paper constrain the data involved to a subspace of interest. The result groups data by month and customer region and reports the sum of store sales.
Second, we want to contrast the result with similar, peer results. To this end, we use the combination of grouper and filter conditions to obtain \textit{sibling} subspaces of the data that query subspaces sharing the same ancestor values with the filters of the original query.
In our case, since we want to constrain time to 1997-Q3, we need to find similar quarters to 1997-Q3 to contrast the results to. This is where the hierarchical nature of the multidimensional space kicks in, by allowing us to define as siblings the values sharing the same ancestor value: assuming that we group quarters in years, we contrast the original query to the quarters of the year 1997, to which 1997-Q3 belongs. Similarly, to find the peers of CA, we contrast CA to the rest of the states of the United States.
Finally, we explore the data space in lower levels of granularity via what we call \textit{drill-down queries} (very much in the traditional OLAP sense), and report the sum of store sales for the months of the quarter of 1997-Q3 and the customers residing in CA. 
We would like to stress here that the emphasis is on the algebraic, operator side and not on the language side: an ANALYZE query could be produced via a CLI call, an SQL-like expression, or a point-n-click interface, yet they would all end-up in the execution of the same algebraic operator.

Apart from introducing formal query semantics for the operator, however, we have been able to theoretically prove that we can obtain the exact same result by merging the facilitator cube queries into a smaller number of merged queries that exploit cube usability results to reduce redundant computation. This theoretical result effectively allows us to introduce a \textit{multi-query optimization (MQO)} strategy that merges several collaborator queries into one or more merged queries, in an attempt to speed up the execution of an ANALYZE query. 
In fact, we have devised several strategies of Multi-Query Optimization at merging the underlying facilitator queries into one. In the simplest case, no optimization is performed and the five facilitator queries are executed independently: we refer to this as \textit{Min-MQO}, as it introduces the least degree of query merging. Second, we follow the theoretical result and merge all the facilitator queries into a single one, a strategy that we call \textit{Max-MQO}, as it implements the original theoretical result of merging everything. Finally, based on our original experimental observations that identified cases of low performance for the aforementioned strategies, we also introduce an intermediate strategy, \textit{Mid-MQO}, that strikes a balance between the two extremes and merges only a subset of the facilitator queries.  In all merging strategies, once the results of the merged queries are obtained, they are post-processed, such that the exact facilitator queries are populated correctly. 

We have extensively evaluated the performance of each MQO strategy via various query workloads on  multiple datasets, to assess both the effect of data size and query selectivity to the execution cost as well as the optimal strategy. Our experiments demonstrate that \textit{Mid-MQO} is the most efficient algorithm that scales smoothly with data size and selectivity and typically achieves the best performance. Although in many cases \textit{Max-MQO} has the worst performance, when the sibling queries are sizable and significantly overlap, \textit{Max-MQO} is the fastest algorithm; we can predict via a decision tree when this is the case. Finally, \textit{Min-MQO} is never the optimal strategy, although it follows \textit{Mid-MQO} fairly closely.

The remainder of this paper is structured as follows. Section \ref{sec:rw} reviews related work in exploratory data analysis and pattern discovery. Section \ref{sec:background} introduces the formal modeling background. Section \ref{sec:analyze} presents the design and semantics of the ANALYZE operator. Section \ref{sec:MQO} develops our multi-query optimization strategies. Section \ref{sec:exps} reports our experimental evaluation. Section \ref{sec:end} concludes with a discussion of the findings and future directions.

\section{Related Work}\label{sec:rw}

\textbf{Automated EDA.} 
Exploratory Data Analysis (EDA) \cite{DBLP:conf/sigmod/MiloS20} allows data analysts to interact with dataset management tools to gain highlights. The goal of EDA is the production of highlights, in order to find interesting, surprising and important facts of a data subspace (likely a query result) and present them using data narration methods \cite{DBLP:conf/vldb/Sarawagi99, DBLP:conf/edbt/SarawagiAM98, DBLP:conf/vldb/Sarawagi00, DBLP:conf/vldb/SatheS01, DBLP:conf/sigmod/IdreosPC15, DBLP:conf/sigmod/TangHYDZ17, DBLP:conf/cikm/ElMS19, DBLP:conf/sigmod/DingHXZZ19, DBLP:journals/tvcg/WangSZCXMZ20, DBLP:conf/sigmod/00040HZ21, DBLP:conf/cikm/PersonnazABFS21, 
DBLP:journals/tvcg/ShiXSSC21, 
DBLP:journals/cacm/BieRHHSW22, DBLP:conf/edbt/Amer-YahiaMP23, 
DBLP:journals/tvcg/SunCCWSC23, DBLP:conf/chi/LiYZWQW23, DBLP:conf/emnlp/MaDWHZ23, 
DBLP:journals/pvldb/AmerYahia24, DBLP:journals/pvldb/XingWJ24, DBLP:conf/edbt/LipmanMSWZ25}. Over the years, several operators have been proposed to complement the fundamental ones. The \emph{DIFF} operator \cite{DBLP:conf/vldb/Sarawagi99} returns the set of tuples that most successfully describe the difference of values between two cells of a cube that are given as input. The same author also describes a method that profiles the exploration of a user and uses the Maximum Entropy principle to recommend which unvisited parts of the cube can be the most surprising in a subsequent query \cite{DBLP:conf/vldb/Sarawagi00}. Finally, the \emph{RELAX} operator allows to verify whether a pattern observed at a certain level of detail is present at a coarser level of detail too \cite{DBLP:conf/vldb/SatheS01}.

Alternative operators have also been proposed in the Cinecubes method \cite{DBLP:conf/dolap/GkesoulisV13, DBLP:journals/is/GkesoulisVM15}. The goal of this effort is to facilitate automated reporting, given an  original OLAP query as input. To achieve this purpose two operators (expressed as \emph{acts}) are proposed, namely, (a) \emph{put-in-context}, i.e., compare the result of the original query to query results over similar, sibling values; and (b) \emph{give-details}, where  drill-downs of the original query's groupers are performed.

\textbf{Intentional Model}. The Intentional Model \cite{DBLP:conf/dolap/VassiliadisM18, DBLP:journals/is/VassiliadisMR19} envisions Business Intelligence tools, where users will apply intentional operators over data, to simplify the querying step of data analysis operations. The user expresses high-level requirements, like 'analyze', 'assess', 'predict', that have to be addressed via auxiliary queries, ML models, highlights and data stories. (a) auxiliary queries that collect related data; (b) models, in the ML sense of data that extract patterns from the collected data, (c) highlights, interesting parts of the data and models, ranked with respect to data interestingness, and (d) data stories, presenting the findings via data visualization and storytelling methods. Although the general framework of the intentional model is specified in  \cite{DBLP:conf/dolap/VassiliadisM18, DBLP:journals/is/VassiliadisMR19}, the exact implementation of the operators is an open problem (to which we currently contribute with respect to the ANALYZE operator).
\cite{DBLP:journals/tkde/FranciaGMRV23, DBLP:conf/edbt/FranciaGMRV21} explore the ASSESS operator for contrasting query results to specified 'benchmarks' of expected performance.
\cite{DBLP:journals/isf/FranciaMPR22} presents the DESCRIBE operator which generates cubes annotated with model components (e.g., clusters, outliers).
\cite{DBLP:journals/is/FranciaRM24} suggests an EXPLAIN operator, which uses statistical models to explain why a measure takes certain values in relation to other measures.
The most similar work to the current one, and the root of the Intentional Model, is the Cinecubes method \cite{DBLP:conf/dolap/GkesoulisV13, DBLP:journals/is/GkesoulisVM15}, which produces a data story as a PowerPoint presentation with results of auxiliary drill-down and sibling queries. 

\textbf{Multi-Query Optimization (MQO)}.\cite{DBLP:journals/tods/Sellis88} establishes the foundational evidence that processing multiple queries together yields considerable cost reductions compared to independent query execution. 
Sellis presents algorithms that search the space of possible combined execution plans, trading off the cost of materializing shared subplans against the benefits of reused computation, and develops heuristics to keep the evaluation cost manageable for realistic workloads. Following up on the aforementioned work,
\cite{DBLP:journals/tkde/SellisG90} provides one of the first formal treatments of MQO, as the task of generating an optimal execution strategy for a set of concurrent queries by identifying and exploiting common subexpressions, shared join paths, and reusable intermediate results. Cost-model extensions are introduced for evaluating shared plans, enabling the optimizer to reason about materialization and reuse.
\cite{DBLP:reference/db/RoyS09} discusses a comprehensive experimental validation for Volcano optimizers. \cite{DBLP:conf/edbt/HongRKGD09, DBLP:conf/icde/LeKDL12} provide extra rewriting techniques. \cite{DBLP:journals/corr/KathuriaS15} provides approximate optimization techniques.
\cite{DBLP:reference/db/RoyS09} \rem{to encycl. lemma ti akribws kanei??} discusses comprehensive experimental validation that demonstrates significant performance improvements with acceptable optimization overhead through the proposed Volcano-based heuristic algorithms. Volcano-based algorithms are designed to make MQO computationally efficient and extensible to complex queries and cost models. The work discusses algorithmic primitives that decompose MQO into tractable subproblems (e.g., bounded-cost search, grouping of queries by similarity) and combines them into an optimizer that scales to larger workloads compared to prior exhaustive approaches. \cite{DBLP:conf/edbt/HongRKGD09} proposes a rule-based approach that takes advantage of declarative transformation rules to detect and restructure shared computations between queries. The rule-based design merges equivalent subexpressions, factors out common operators, and introduces shared materialized results where cost.  
\cite{DBLP:journals/pvldb/WangC13} reports performance improvements over state-of-the-art techniques in MapReduce frameworks. The authors formulate techniques to identify common computations across multiple MapReduce jobs. In this context, MapReduce jobs are rewritten or scheduled so that shared map or reduce stages are executed once, and their outputs are reused across jobs. 
\cite{DBLP:conf/icde/LeKDL12} introduces techniques that identify common subgraph patterns, such as join motifs, and leverages them to build unified execution plans capable of serving many SPARQL queries simultaneously. The authors propose structural indexing and graph-pattern clustering methods to reduce the combinatorial explosion in multi-query plan enumeration. 
Beyond performance gains, \cite{DBLP:journals/corr/KathuriaS15} advances the theoretical foundation by providing the first approximation factor guarantees for MQO algorithms. The authors propose cost-aware grouping and plan-construction techniques that exploit structural properties of query workloads, enabling the discovery of beneficial shared subexpressions without exhaustive enumeration.

\textbf{Comparison to the State of the Art}. The main contribution of this work is the formal introduction of a novel operator along with its optimization techniques. Although conceptually related, our work here does not pertain to the core of the area of multi-query optimization: our MQO algorithms do not explore alternative query execution plans, but rather exploit theoretical guarantees of correctness.  Compared to the most similar line of work, \cite{DBLP:conf/dolap/GkesoulisV13, DBLP:journals/is/GkesoulisVM15}, we change the query semantics to address efficiency issues (thus we face a new problem), provide formal semantics of an operator and, and most importantly, we introduce multi-query optimization strategies for the speed up of query processing. Compared to the rest of the corpus of related literature, to the best of our knowledge, there is no other work with multi-query optimization strategies for a formal, single operator that provides a $360^{o}$ view of the data with opportunities for optimization exactly due to the integration of various queries. 

\silence{\textbf{Comparison}. Compared to \cite{DBLP:conf/dolap/GkesoulisV13, DBLP:journals/is/GkesoulisVM15} we change the semantics of the internal queries, and, most importantly, we introduce multi-query optimization strategies for the further speed up of the execution. Compared to the rest of the corpus of the related literature, to the best of our knowledge, there is no other work with multi-query optimization strategies for an operator that provides a $360^{o}$ view of the data.}

\begin{figure}
  \begin{center}
    \includegraphics[width=0.8\textwidth]{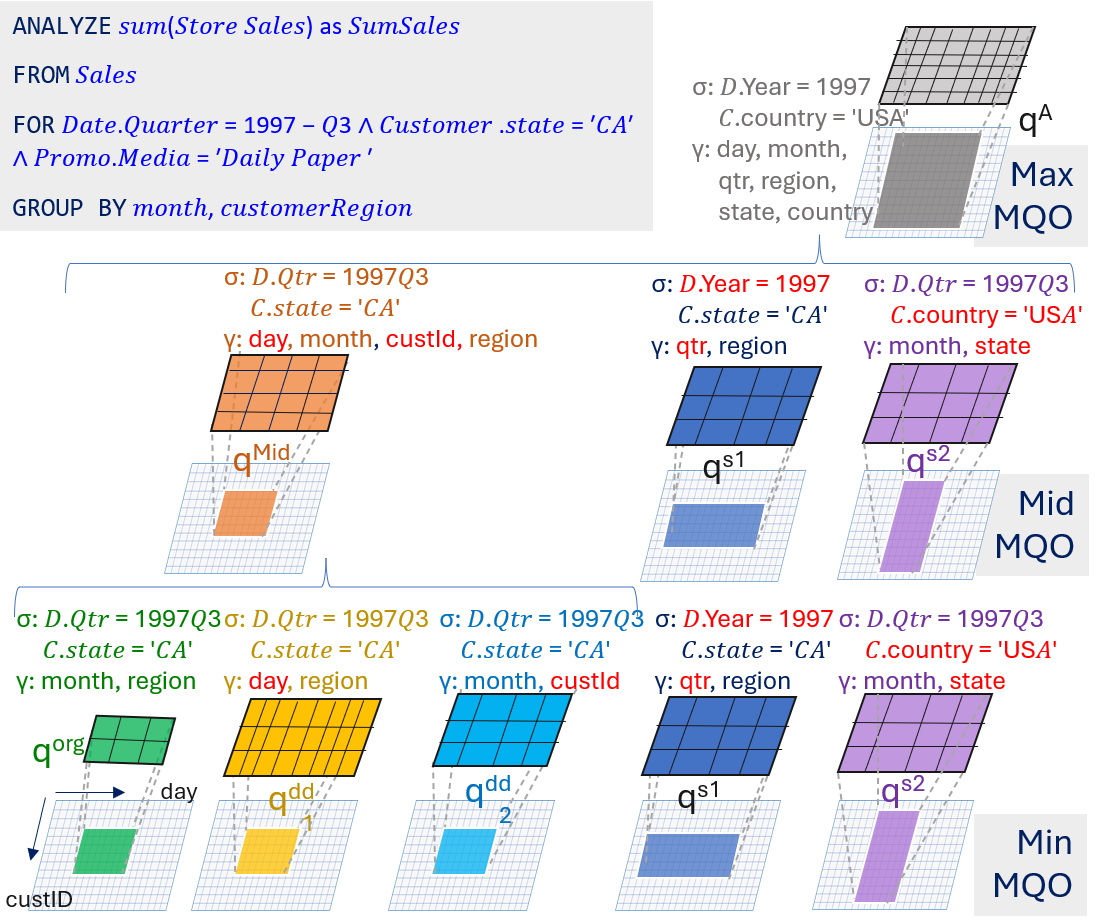}
  \end{center}
  \caption{Facilitator queries per algorithm}
      \label{fig:refFig}
\end{figure}

\section{Preliminaries \& Formal Background}\label{sec:background}

In our deliberations, we assume the formal model of \cite{PV22} (also used in \cite{DBLP:conf/dolap/Vassiliadis23}) for the definition of the multidimensional space, cubes and cube queries. 

\silence{
We follow a simplified apodosis of the formalities here to allow a concise description. The reader is assumed to have knowledge of fundamental OLAP concepts. \sticky{Carbon DOLAP23}
}

\subsection{Multidimensional Space}

\textbf{Multidimensional space}. Data are defined in the context of a multidimensional space. The multidimensional space includes a finite set of dimensions. Dimensions  provide the context for factual measurements and will be structured in terms of dimension levels, which are 
abstraction levels that aid in observing the data at different levels of granularity. For example, the dimension $Time$ is structured on the basis of the dimension levels $Day$, $Month$, $Year$, $All$. A \textit{dimension level} $L$ includes a name and a finite set of values, $dom$($L$), as its domain. Following the traditional OLAP terminology, the values that belong to the domains of the levels are called \textit{dimension members}, or simply \textit{members} (e.g., the values Paris, Rome, Athens are members of the domain of level $City$, and, subsequently, of dimension $Geography$).
A \textit{dimension} is a non-strict partial order of a finite set of \textit{levels}, obligatorily including (a) a most detailed level at the lowest possible level of coarseness, and (b) an upper bound, which is called $ALL$, with a single value `All'. We denote the partial order of dimensions with $\preceq$, i.e., $D.L_{low}$ $\preceq$ $ D.L_{high}$ signifies that $D.L_{low}$ is at a lower level of coarseness than $D.L_{high}$ in the context of dimension $D$ -- e.g., $Geo.City$ $\preceq$ $Geo.Country$. To ease notation, unless explicitly mentioned otherwise, in the sequel, we will assume total orders of levels, which means that there is a linear order of the levels starting from the lower level and ending at ALL with a linear chain of precedence.

We can map the members at a lower level of coarseness to values at a higher level of coarseness via an \textit{ancestor function} $anc_{L^l}^{L^h}()$. Given a member of a level $L_{l}$ as a parameter, say $v_{l}$, the function $anc_{L^l}^{L^h}()$ returns the corresponding ancestor value, for $v_{l}$, say $v_{h}$, at the level $L_{h}$, i.e., $v_{h}$ = $anc_{L^l}^{L^h}(v_{l})$. The inverse of an ancestor function is not a function, but a mapping of a high level value to a set of \textit{descendant values} at a lower level of coarseness (e.g., $Continent$ Europe is mapped to the set of all European cities at the $City$ level), and is denoted via the notation $desc_{L^h}^{L^l}()$. For 
example $Europe$ = $anc_{City}^{Continent}(Athens)$. See \cite{PV22} for more constraints and explanations. The union of the domains of all levels of a dimension will be abusively called the domain of the dimension, $dom(D)$. The multidimensional space is generated from the Cartesian product of the domains of its dimensions, thus each point c in the multidimensional space carries the coordinates $c(c_{1}, \ldots, c_{n})$, with exactly one value for each dimension. 

\textbf{Cubes}. Facts are structured in cubes. A \textit{cube} $C$ is defined with respect to several dimensions, fixed at specific levels and also includes a number of \textit{measures} to hold the measurable aspects of its facts. Thus the schema of a cube is a set of attributes, including a set of dimension levels (over different dimensions) and a set of measures that include factual measurements for the data stored in the cube.

\sloppy
Formally, the schema of a cube $schema(C)$, is a tuple, say $[D_{1}.L_{1},..., D_{n}.L_{n},,M_{1}, ..., M_{m}]$, or simply $[L_{1}, ..., L_{n}, M_{1}, ..., M_{m}]$, with the combination of the dimension levels (each coming from a different dimension) acting as primary key and context for the measurements and a set of measures as placeholders for the (aggregate) measurements. 

If all the dimension levels of a cube schema are the lowest possible levels of their dimension, the cube is a 
\textit{detailed cube}, typically denoted via the notation $C^{0}$ with a schema $[D_{1}.L^{0}_{1}, ..., D_{n}.L^{0}_{n}, M^{0}_{1}, ..., M^{0}_{m}]$. The result of a query $q$ is a set of cells that we denote as $q.cells$.
Each record of a cube $C$ under a schema
$[D_1.L_1, \dots, D_n.L_n$, $M_1, \dots, M_m]$, also known as a \textit{cell}, is a tuple $c$ = $[l_{1}, \ldots , l_{n}$, $m_{1}, \ldots ,$ $m_{m}]$,  such that $l_{i} \in dom(D_{i}.L_{i})$ and $m_{j}$ $\in$ $dom(M_{j})$. The vector $[l_{1}, \ldots , l_{n}]$ signifies the \textit{coordinates} of a cell. Equivalently, a $cell$ can be thought as a point in the multidimensional space of the cube's dimensions annotated with, or hosting, a set of measures. A cube includes a finite set of cells as its extension.

Overall, a cube database includes both hierarchies and cubes as distinct but related first-class citizens. A cube is defined with respect to levels that belong to the hierarchies, and, therefore, its position in the multidimensional hierarchical space constructed by the hierarchies is determined by these levels. 

\subsection{Cube Queries}

\textbf{Queries}. A cube query is a cube too, specified by: (a) the detailed cube over which it is imposed, (b) a selection condition that isolates the facts that qualify for further processing, (c) the grouping levels, which determine the 
coarseness of the result, and, (d) an aggregation over some or all measures of the cube that accompanies the grouping levels in the final result.

\begin{center}
    $q$ = $<C^{0}$, $\phi$, $[L_{1}, ..., L_{n},
M_{1}, ..., M_{m}]$, $[agg_{1}(M^{0}_{1} )$, ...,$agg_{m}(M^{0}_{m})]$ $>$
\end{center}



We assume selection conditions which are conjunctions of atomic filters of the form $L$ = $value$, or in general $L~\in~\{v_1,\dots, v_k\}$. 

Selection conditions of this form can eventually be translated to their \textit{equivalent} selection conditions at the detailed level, via the conjunction of the detailed equivalents of the atoms of $\phi$. Specifically, assuming an atom $L~\in~\{v_1,\dots, v_k\}$, then $L^{0}$ $\in$ $\{desc_L^{L^0}(v_{1}) \cup ... \cup desc_L^{L^0}(v_{k})\}$, eventually producing an expression $L^{0}$ $\in$ $\{v'_{1} ,..., v'_{k'}\}$ is its detailed equivalent, called \textit{detailed proxy}. The reason for deriving $\phi^{0}$ is that $\phi^{0}$, as the conjunction of the respective atomic filters at the most detailed level, is directly applicable over $C^{0}$ and produces exactly the same subset of the multidimensional space as $\phi$, albeit at a most detailed level of granularity. For example, assume $Year~\in~\{2018,2019\}$, its detailed proxy is $Day~\in~\{2018/01/01, \ldots, 2019/12/31\}$. For a dimension $D$ that is not being explicitly filtered by any atom, one can equivalently assume a filter of the form $D.ALL~=~all$.


We also assume aggregation functions $agg_{i}$ include aggregate functions like $\{sum,max, min, count, ...\}$ with the respective well-known semantics.


\textbf{Cube Query Semantics}. The semantics of the query are:
(i) apply $\phi^0$, the detailed equivalent of the selection condition over $C^{0}$ and produce a subset 
of the detailed cube, say $q^{0}$, known as the
detailed area of the query;
(ii) map each dimension member to its ancestor value at the level specified by the grouping levels and group the tuples with the same coordinates in the same same-coordinate group;
(iii) for each same-coordinate group, apply the aggregate functions to the measures of its cells, thus producing a single value per aggregate measure.

A cube query is also a cube under the schema $[L_{1}, ..., L_{n},M_{1}, ..., M_{m}]$ and with the cells of the query result (denoted as $q.cells$) as its extension.

\textbf{Constraints}. In the context of this work, we imply several assumptions that mainly serve the purpose of clarity. First, we work with cube queries that involve a single measure. Second, we assume strictly two aggregator levels for the result, and third, for all atomic filters, we assume that if the dimension with a filter is also a grouper, the atomic filter is expressed in a level greater than or equal to the grouper.

\begin{itemize}
\item We work with cube queries that involve a single measure 
\item We assume strictly two aggregator levels for the result 
\item For all the atomic filters, we assume that, if the dimension with a filter is also a grouper, the atomic filter is expressed in a level greater or equal than the grouper. We will denote with $D.L^{\sigma}$ the level of the atomic filter ($D.L^{\sigma} = v$) and with $D.L^{\gamma}$ the respective grouper, and require that if such a pair exists in the query definition, then $D.L^{\gamma}$ $\preceq$ $ D.L^{\sigma}$
\end{itemize}
The rationale behind the above constraints is as follows:
\begin{itemize}
    \item To avoid overloading notation, we opt to work with a single measure for the purpose of clarity.It is straightforward to generalize our results to multiple measures by joining single-measured cube query results over their dimension levels.
    \item We also opt to work with exactly 2 grouper levels. This is a traditional, reasonable constraint that stems from the presentational constraint of having to present results in a 2D screen, either in a tabular form, or a chart.
    \item The constraint on the relationship of filters and groupers refer to the notion of perfect rollability of \cite{PV22}. For example, assuming we are grouping results by $city$, if a filter is applied over the $Geography$ dimension it should involve a specific $city$, or a $country$ (whose measures we group by city), or a $continent$. Had we filtered for a certain $neighborhood$, and then requested to group by city would (a) produce results, although (b) with at least unclear semantics, and, in any case (c) violating the perfect rollability requirement. Hence, we believe the constraint is reasonable and useful. \rem{silenced signatures}
\end{itemize}

Then, we will employ the following query expression:
\begin{center}
    $q$ = $<C^{0}$, $\phi$, $[L_{\alpha}, L_{\beta},M]$, $agg(M^{0})$ $>$
\end{center}

\section{The Analyze Operator}\label{sec:analyze}
In this section, we introduce the \ANALYZE operator in terms of semantics, syntax and execution strategies.

\subsection{Syntax \& Semantics}
Assume the setup already discussed in Section~\ref{sec:background}, with a detailed cube $C^0$ defined over a list of hierarchical dimensions. The syntax of the \ANALYZE operator is as follows: 

    \begin{tabular}{ll}
        \textsc{analyze} & \color{blue}{$agg(M^0)$ as $M$ \textsc{from} $C^0$}\\ 
        \textsc{for} & \color{blue}{$\phi$} \\
        \textsc{group by} & \color{blue}{$L_{\alpha}, L_{\beta}$}\\
        \textsc{as} &\color{blue}{$queryName$} \\
    \end{tabular}

\noindent with the variables participating in the definition having the obvious semantics already introduced in Section~\ref{sec:background}.

To complement this SQL-like definition, we will also employ an algebraic representation of the operator, as 
$analyze (C^{0}, \phi, [D_{\alpha}.L_{\alpha}, D_{\beta}.L_{\beta},M],agg(M^0) )$, or equivalently as:

\begin{center}
    $analyzeOp$ = $\langle$ $C^{0}$, $\phi$, [$D_{\alpha}.L_{\alpha}$,      $D_{\beta}.L_{\beta}$,$M$], $agg(M^0)$ $\rangle$
\end{center}

\textbf{Example}. To exemplify our discussions, we will employ a cubefied version of 
the well-known Foodmart schema \cite{Foodmart}. Foodmart includes a cube, \textsf{Sales} with 3 measures, contextualized by 5 dimensions, namely \textsf{Product, Date, Promotion, Customer}, and \textsf{Store}, each with its own levels.
All dimensions include the level \textsf{All} as the uppermost level of coarseness.
The query presented in Section~\ref{sec:intro} will be used as our reference example.

\textbf{Semantics}. The semantics of the operator involve the introduction of three separate groups of intermediate cube queries, which we call facilitator queries, with each group producing a set of distinct cube query results. Assuming a given \ANALYZE query $AQ$, the query set of $AQ.queries$ is a triplet of query sets $< Q^{org}, Q^{sib}, Q^{dd} >$ 
and the result of $AQ$, $AQ.result$ is a triplet with the results of the respective queries in the three query sets.

The 3 query sets that the operator produces are as follows:
\begin{itemize}
    \item
    (a) $Q^{org}$ is a singleton set, comprising exactly one query $q^{org}$, to which we will refer as \textit{the original query}, that computes the aggregate value for $M^0$ as specified by the variables of the operator. 
    \item 
    (b) $Q^{sib}$ is a set of sibling queries, whose goal is to compare the behavior of key values in the operator definition against the behavior of their peer values.
    \item
    (c) $Q^{dd}$ is a set of drill-down queries, whose goal is to provide more detailed data to the analyst for the values produced by the original query.
\end{itemize}
Figure~\ref{fig:refFig} demonstrates the queries of the reference example, along with the auxiliary queries that the subsequent MQO algorithms will produce. In the sequel, we present the semantics and rationale of each of these auxiliary, facilitator queries.

\subsubsection{Sibling queries}

The key word concerning siblings is \textit{context}:
Sibling queries aim to compare (a) the different slices of the data space that pertain to the original query against (b) their peers, in order to put them in context. By comparing data slices to their peer, the data slices are effectively assessed by the analyst and, thus, contextualized.

\textbf{Original idea}. Intuitively, given an original query, a sibling query differs from it as follows: instead
of an atomic selection formula $L_{i}$ =$v_{i}$ of the original query, the sibling query
contains a formula of the form $L_{i}$ $\in$
$children(parent(v_{i}))$. 

\hrulefill

Formally, assume an original query $q^{org}$

\begin{center}
$q^{org}$ = $\langle$ $C^{0}$, $\phi_{\alpha}$ $\wedge$ $\phi_{\beta}$ $\wedge$ $\phi_{\Box}$, [$D_{\alpha}.L_{\alpha}^{\gamma}$, $D_{\beta}.L_{\beta}^{\gamma}$,$M$],
$agg(M^0)$ $\rangle$
\end{center}

\noindent with the following constraints:
\begin{itemize}
    \item $\phi_{\alpha}$ an optional selection atom of the form: $D_{\alpha}.L_{\alpha}^{\sigma}$ = $v_{\alpha}$, $L_{\alpha}^{\gamma}$ $\preceq$ $L_{\alpha}^{\sigma}$
    \item $\phi_{\beta}$ an optional selection atom of the form: $D_{\beta}.L_{\beta}^{\sigma}$ = $v_{\beta}$, $L_{\beta}^{\gamma}$ $\preceq$ $L_{\beta}^{\sigma}$
    \item $\phi_{\Box}$ an optional conjunction of atoms for dimensions other than the groupers $D_{\alpha}$, $D_{\beta}$, with all its atoms $\phi_{i}$ being of the form, $\phi_{i}$: $L_{i}$ = $v_{i}$
\end{itemize}
Observe that for each grouper dimension we discriminate between the grouper level (e.g., $L_{\beta}^{\gamma}$) and the filtering level, e.g., $L_{\beta}^{\sigma}$. In our example, we group by $L_{\beta}^{\gamma}$: $state$, whereas our $L_{\beta}^{\sigma}$ level is $country$ ($country~=~'USA'$), at a higher level than the grouper.

In its simplest possible form, a sibling query of $q^{org}$ with respect to dimension say $D_{\alpha}$, would be

\begin{center}
$q^{\color{blue}{s^\alpha}}$ = $\langle$ $C^{0}$, {\color{blue}{$\phi_{\alpha}^{\star}$}} $\wedge$ $\phi_{\beta}$ $\wedge$ $\phi_{\Box}$, [$D_{\alpha}.L_{\alpha}^{\gamma}$, $D_{\beta}.L_{\beta}^{\gamma}$,$M$],
$agg(M^0)$ $\rangle$ 
\end{center}

\noindent with $\phi_{\alpha}^{\star}$ an atom of the form: $D_{\alpha}.L_{\alpha}^{\sigma}$ = {\color{blue}{$v_{\alpha}^{\star}$}}, under the constraint that $anc_{L_\alpha^\sigma}^{L_{\alpha+1}^\sigma}(v_\alpha)$ = $anc_{L_\alpha^\sigma}^{L_{\alpha+1}^\sigma}(v_\alpha^\star)$.

\hrulefill

In other words, we replace the selection value of exactly one atom with a sibling value (expressed by the constraint that their mother-value, one level higher, is the same). In our example, if we replaced the constraint $quarter~=~1997-Q3$ with $year~=~1997-Q4$, retaining the rest of the query untouched, this creates a sibling query for the Date dimension. See \cite{DBLP:journals/is/GkesoulisVM15} for a detailed discussion of the general case.

\textbf{Definition of a convenient variant for the sibling computation}. A clear problem here is that if a value used in a selection atom has too many sibling values, this produces a large number of queries, even if we create siblings only for just this dimension. The problem generalizes if we take all selection atoms into consideration. To avoid this complexity we simplify sibling generation by (i) considering siblings only for the atoms $\phi_{\alpha}$ and $\phi_{\beta}$, and, (ii) by merging all sibling slices into one, by slightly adapting the aggregation level.

\hrulefill

Given the original query $q^{org}$

\begin{center}
$q^{org}$ = $\langle$ $C^{0}$, $\phi_{\alpha}$ $\wedge$ $\phi_{\beta}$ $\wedge$ $\phi_{\Box}$, [$D_{\alpha}.L_{\alpha}^{\gamma}$, $D_{\beta}.L_{\beta}^{\gamma}$,$M$],
$agg(M^0)$ $\rangle$ 
\end{center}

we generate two sibling queries, each for one of the groupers as:

\begin{center}
$q^{s^\alpha}$ = $\langle$ $C^{0}$, {\color{blue}{$\phi_{\alpha}^{\star}$}} $\wedge$ $\phi_{b}$ $\wedge$ $\phi_{\Box}$, [{\color{blue}{$D_{\alpha}.L_{\alpha}^{\sigma}$}}, $D_{\beta}.L_{\beta}^{\gamma}$,$M$],
$agg(M^0)$ $\rangle$, 
{\color{blue}{$\phi_{\alpha}^{\star}$: $L_{\alpha+1}^{\sigma}$ = $anc_{L_\alpha^\sigma}^{L_{\alpha+1}^\sigma}(v_\alpha)$}}
\end{center}

\begin{center}
$q^{s^\beta}$ = $\langle$ $C^{0}$, $\phi_{\alpha}$ $\wedge$ {\color{blue}{$\phi_{\beta}^{\star}$}} $\wedge$ $\phi_{\Box}$, [$D_{\alpha}.L_{\alpha}^{\gamma}$, {\color{blue}{$D_{\beta}.L_{\beta}^{\sigma}$}} ,$M$],
$agg(M^0)$ $\rangle$, 
{\color{blue}{$\phi_{\beta}^{\star}$: $L_{\beta+1}^{\sigma}$ = $anc_{L_\beta^\sigma}^{L_{\beta+1}^\sigma}(v_\beta)$}}
\end{center}

\hrulefill
Intuitively, we perform the following two alterations to the original query:
\begin{itemize}
    \item Instead of asking for the value of the filtering level of the dimension under moderation, we ask for its mother value. So, instead of a $state~=~CA$ filter, we introduce a filter $mama(state)~=~mama(CA)$, i.e., $country~=~USA$.
    \item We aggregate by the former filtering level. In this example, now that we have selected as the entire $USA$ as the subspace to work with, we group by the former selector level, $state$, thus producing states as the grouper values of the result and effectively comparing all the sibling states in $USA$ to one another (including our reference one, $CA$).
\end{itemize}

In our example:
\begin{center}
$q^{org}$ = $\langle$ $Sales$, $Date.Quarter~=~1997Q3$ $\wedge$ $Customer.State~=~'CA'$ $\wedge$ $Promo.media~=~'Daily~Paper'$, [$Date.Month$, $Customer.Region$, $MORG$],
$sum(StoreSales)$ $\rangle$ 
\end{center}

\begin{center}
$q^{s^{\color{blue}{Date}}}$ = $\langle$ $Sales$, {\color{blue}{$Date.Year~=~'1997'$}} $Customer.State~=~'CA'$ $\wedge$ $Promo.media~=~'Daily~Paper'$, 
[{\color{blue}{$Date.Quarter$}}, $Customer.Region$, $MSDATE$],
$sum(StoreSales)$ $\rangle$ 
\end{center}

\begin{center}
$q^{s^{\color{blue}{Cust}}}$ = $\langle$ $Sales$, $Date.Quarter~=~1997Q3$ $\wedge$  {\color{blue}{$Customer.Country~=~'USA'$}} $\wedge$ $Promo.media~=~'Daily~Paper'$, [$Date.Month$, {\color{blue}{$Customer.State$}},$MSCUST$],
$sum(StoreSales)$ $\rangle$ 
\end{center}

\subsubsection{Drill-Down queries}

Apart from contextualizing the result of the original query against its peers, we can also provide further details for the displayed data, by drilling into the dimension hierarchies of the produced results. Ideally, this requires drilling into each of the cells of the results of the original query -- see \cite{DBLP:journals/is/GkesoulisVM15} on how this was traditionally done. However, this produces a significantly large number of queries to execute. We adopt a simple solution to this problem by drilling into each of the aggregator levels, by going down one level in their hierarchy. Thus, we produce two drill down queries, each drilling a level down into the detail of one of the grouper dimensions.

\hrulefill

Given the original query $q^{org}$

\begin{center}
$q^{org}$ = $\langle$ $C^{0}$, $\phi$, [$D_{\alpha}.L_{\alpha}^{\gamma}$, $D_{\beta}.L_{\beta}^{\gamma}$,$M$],
$agg(M^0)$ $\rangle$ 
\end{center}

we generate two drill-down queries, each for one of the groupers as:

\begin{center}
$q^{dd^{\color{blue}{\alpha}}}$ = $\langle$ $C^{0}$, $\phi$, 
[{\color{blue}{$D_{\alpha}.L_{\alpha-1}^{\gamma}$}}, 
$D_{\beta}.L_{\beta}^{\gamma}$,$M^{dd^{\color{blue}{\alpha}}}$],
$agg(M^0)$ $\rangle$ 
\end{center}

\begin{center}
$q^{dd^{\color{blue}{\beta}}}$ = $\langle$ $C^{0}$, $\phi$, 
[$D_{\alpha}.L_{\alpha}^{\gamma}$, 
{\color{blue}{$D_{\beta}.L_{\beta-1}^{\gamma}$}},$M^{dd^{\color{blue}{\beta}}}$],
$agg(M^0)$ $\rangle$ 
\end{center}

\hrulefill \\
In our example:
\begin{center}
$q^{org}$ = $\langle$ $Sales$, $Date.Quarter~=~1997Q3$ $\wedge$ $Customer.State~=~'CA'$ $\wedge$ $Promo.media~=~'Daily~Paper'$, 
[$Date.Month$, $Customer.Region$, $MORG$],
$sum(StoreSales)$ $\rangle$ 
\end{center}

\begin{center}
$q^{s^{\color{blue}{Date}}}$ = $\langle$ $Sales$, {\color{blue}{$Date.Year~=~'1997'$}} $Customer.State~=~'CA'$ $\wedge$ $Promo.media~=~'Daily~Paper'$, 
[{\color{blue}{$Date.Quarter$}}, $Customer.Region$, $MSDATE$],
$sum(StoreSales)$ $\rangle$ 
\end{center}

\begin{center}
$q^{s^{\color{blue}{Cust}}}$ = $\langle$ $Sales$, $Date.Quarter~=~1997Q3$ $\wedge$  {\color{blue}{$Customer.Country~=~'USA'$}} $\wedge$ $Promo.media~=~'Daily~Paper'$, [$Date.Month$, {\color{blue}{$Customer.State$}},$MSCUST$],
$sum(StoreSales)$ $\rangle$ 
\end{center}

\begin{center}
$q^{dd^{\color{blue}{Date}}}$ = $\langle$ $Sales$, $Date.Quarter~=~1997Q3$ $\wedge$ $Customer.State~=~'CA'$ $\wedge$ $Promo.media~=~'Daily~Paper'$, 
[{\color{blue}{$Date.Day$}},
  $Customer.Region$, $MDDDate$],
$sum(StoreSales)$ $\rangle$ 
\end{center}

\begin{center}
$q^{dd^{\color{blue}{Cust}}}$ = $\langle$ $Sales$, $Date.Quarter~=~1997Q3$ $\wedge$ $Customer.State~=~'CA'$ $\wedge$ $Promo.media~=~'Daily~Paper'$, 
[$Date.Month$,
  {\color{blue}{$Customer.CustomerId$}}, $MDDCust$],
$sum(StoreSales)$ $\rangle$ 
\end{center}


\section{Multiple Query Optimization for Analyze Queries}\label{sec:MQO}
So far, we have seen that the \ANALYZE operator is actually a composition of 5 facilitator queries to the underlying database. In this section, we introduce a method that replaces the need of issuing five queries to the database with a merging of (a subset of) them, along with algorithms that exploit this merging by issuing a single query to (a) collect all the data that the 5 queries require, at an appropriate level of aggregation, and (b) post-process them, in memory, to produce the results of the 5 queries correctly.
This is (a) costly, as the cost of multiple queries adds up the query overheads, and, (b) avoidable, as we will show in the sequel.

The method aims to gain in speed as (a) the query overheads are avoided, (b) multiple requests for the same data are reduced to a single fetch from the database (even, if intra-DBMS caching would have minimized this cost to a certain extent in the naive execution) and, (c) the fact that we are post-processing cells close to what would have been query results means that the amount of information we are post-processing is small -- hence, we can do it efficiently in main memory.

\newcommand{\tuple}[1]{\ensuremath{\left \langle #1 \right \rangle }
}

\subsection{Theoretical Background}
We base our theoretical construction on the Cube Usability Theorem \cite{PV22, DBLP:conf/dolap/Vassiliadis23}, which states that it is possible to derive the result of a cube query by filtering already accessed data by another cube query and re-aggregate them, if certain conditions are held.

For completeness, we repeat here the formulation of the theorem, and redirect the reader to the respective articles for the details of the proof and the supporting concepts.

\hrulefill

\begin{theorem}[Cube Usability (\cite{DBLP:conf/dolap/Vassiliadis23,  PV22})]\label{theor:cubeUsability}
Assume the following two queries: 

\[q^n = \tuple{ \mathbf{DS}^{0},\ \phi^n,\ [L_1^n,\ldots,L_n^n, M_1,\ldots,M_m],\ [agg_1(M_1^0),\ldots,agg_m(M_{m}^0)]\ }\]

and

\[q^b = \tuple{\mathbf{DS}^{0},\ \phi^b,\ [L_1^b,\ldots,L_n^b, M_1,\ldots,M_m],\ [agg_1(M_1^0),\ldots,agg_m(M_{m}^0)] }\]

The query $q^b$ is \emph{usable for computing}, or simply, \emph{usable for} query $q^n$, if the following conditions hold:
\begin{enumerate}[label=\roman*.]
    \item both queries have exactly the same underlying detailed cube $\mathbf{DS}$, 
    \item both queries have exactly the same dimensions in their schema and the same aggregate measures $agg_i(M_i^0)$, $i$ $\in$ 1 .. $m$ (implying a 1:1 mapping between their measures), with all $agg_i$ belonging to a set of known distributive functions. 
    \item both queries have exactly one atom per dimension in their selection condition, of the form $D.L~\in~\{v_1, \ldots, v_k\}$ and selection conditions are conjunctions of such atoms,
    \item both queries have schemata that are perfectly rollable with respect to their selection conditions, which means that grouper levels are perfectly rollable with respect to the respective atom of their dimension, 
    \item all schema levels of query $q^n$ are ancestors (i.e, equal or higher) of the respective levels of $q^b$, i.e., $D.L^{b}$ $\preceq$ $D.L^{n}$, for all dimensions $D$, and,
    \item for every atom of $\phi^n$, say $\alpha^n$, if (i) we obtain $\alpha^{n@L^{b}}$ (i.e., its detailed equivalent at the respective schema level of the previous query $q^b$, $L^b$) to which we simply refer as $\alpha^{n@b}$, and, (ii) compute its signature $\alpha^{{n@b}^{+}}$, then (iii) this signature is a subset of the grouper domain of the respective dimension at $q^b$ (which involves the respective atom $\alpha^b$ and the grouper level $L^b$), i.e., $\alpha^{{n@b}^{+}}$ $\subseteq$ $gdom(\alpha^b, L^{b})$.
\end{enumerate}
\end{theorem}

\hrulefill \\

In our case, we introduce the auxiliary, \textit{all-encompassing query} $q^A$ for the \ANALYZE operator. The all-encompassing query $q^{A}$ is characterized by: (0) the same basic cube and measure aggregation, (1) groupers involving the highest levels of aggregation that are low enough to answer any of the internal queries of the \ANALYZE operator, and, (2) the broadest possible selection condition to encompass all the detailed tuples needed to answer any internal query. In the rest of this section we will first show that the all-encompassing query $q^{A}$ can answer all the internal queries, and, second, we will accompany the feasibility result with two algorithms that actually compute the answer.

\hrulefill

\begin{theorem}[Multi-Query Usability] 
Assume the query 

$analyze$ = $\langle$ $C^{0}$, $\phi_{\alpha}$ $\wedge$ $\phi_{\beta}$ $\wedge$ $\phi_{\Box}$, [$D_{\alpha}.L_{\alpha}^{\gamma}$, $D_{\beta}.L_{\beta}^{\gamma}$,$M$], $agg(M^0)$ $\rangle$

\noindent producing five internal queries as prescribed in the definition of the \ANALYZE operator.
The following \textit{all-encompassing query} $q^{A}$ can answer all the internal queries of the \ANALYZE operator.

\begin{center}
\noindent $q^{A}$=$\langle$$C^{0}$, $\phi_{\alpha}^{\star}\wedge\phi_{\beta}^{\star}\wedge\phi_{\Box}$, 
    [$L_{\alpha-1}^{\gamma}, L_{\beta-1}^{\gamma}, L_{\alpha}^{\gamma}, L_{\beta}^{\gamma}, L_{\alpha}^{\sigma}, L_{\beta}^{\sigma}$,
    $M^A$], $agg(M^0)$$\rangle$,

\noindent with {$\phi_{\alpha}^{\star}$: $L_{\alpha+1}^{\sigma}$ = $anc_{L_\alpha^\sigma}^{L_{\alpha+1}^\sigma}(v_\alpha)$} and
{$\phi_{\beta}^{\star}$: $L_{\beta+1}^{\sigma}$ = $anc_{L_\beta^\sigma}^{L_{\beta+1}^\sigma}(v_\beta)$}
\end{center}
\end{theorem}
\hrulefill
\begin{proof}
Some of the requirements of the Cube Usability Theorem are provided by construction of the operator and the all-encompassing query. We refer the reader to Table~\ref{tab:mqoProof1} as a clear demonstration of why the items in the following list hold. Specifically:
\begin{enumerate}[label=\roman*.]
    \item All queries are on the same detailed cube.
    \item All queries have the same dimensions, the same measure and distributive aggregate function.
    \item All queries have one atom per dimension of the form $D.L = v$ ($true$ or equivalently, $D.ALL = all$ atoms are implied). In fact, the queries have non-trivial atoms for exactly the same dimensions, and even more, they can possibly differ only in the atoms of the grouper dimensions.
    \item All queries have selector levels higher than the grouper levels by construction. We have already assumed that for all dimensions $D$ having $D.L^{\gamma}$ as a grouper level and $D.L^{\phi}$ as the level involved in the selection condition's atom for $D$, $D.L^{\gamma}$ $\preceq$ $D.L^{\phi}$, i.e., the selection condition is defined at a higher level than the grouping.    
    Therefore, $D.L^{\gamma}_{i-1}$ $\preceq$ $D.L^{\gamma}_i$ $\preceq$ $ D.L^{\sigma}_i$
    by construction.
    \item All queries have their grouper levels higher or equal to the ones of the all-encompassing one (again, remember that $D.L^{\gamma}_i$ $\preceq$ $ D.L^{\sigma}_i$).
    \item Apart from the atoms for the grouper dimensions, all selection conditions in all queries are identical  (i.e., they all share the same $\phi_{\Box}$). The grouper dimensions' atoms also fulfill the constraint (vi) as we will show right away.
\end{enumerate}

Therefore, all we need to show is the sixth constraint (iv), stating that once the all-encompassing query has been computed, it is possible to apply a selection condition to its result that will yield the same result as the original query. We start by showing that, given the result of the all-encompassing query which has been computed having applied $\phi_{\alpha}^{\star}$, it is attainable to produce the exact same result with having applied $\phi_{\alpha}$ (the proof is identical for $\beta$).


\begin{table*}[t]
\begin{tabular}{lcccccccc}
                  & $C^0$ 
                  & \multicolumn{3}{c}{$\phi$: $\phi_{\alpha}$ $\wedge$ $\phi_{\beta}$ $\wedge$ $\phi_{\Box}$} 
                  & \multicolumn{3}{c}{Schema }                                    
                  & $agg$    \\
                  
                  & $[i]$ 
                  & \multicolumn{3}{c}{($[iii][iv][vi]$)} 
                  & \multicolumn{3}{c}{($[ii][iv][v]$)}                                    
                  & $[ii]$    \\
\hline
$q^{A}$           
    & $C^{0}$ 
    & $\phi_{\alpha}^{\star}$
    & $\phi_{\beta}^{\star}$
    & $\phi_{\Box}$
    & $D_{\alpha}.L_{\alpha-1}^{\gamma}$ 
    & $D_{\beta}.L_{\beta-1}^{\gamma}$ 
    & $M^{A}$   & $agg(M^0)$ \\
\hline
$q^{org}$         
    & $C^{0}$ 
    & $\phi_{\alpha}$                     
    & $\phi_{\beta}$                    
    & $\phi_{\Box}$                
    & $D_{\alpha}.L_{\alpha}^{\gamma}$   
    & $D_{\beta}.L_{\beta}^{\gamma}$   
    & $M^{org}$ & $agg(M^0)$ \\

$q^{dd^{\alpha}}$ 
    & $C^{0}$ 
    & $\phi_{\alpha}$                     
    & $\phi_{\beta}$                    
    & $\phi_{\Box}$                
    & $D_{\alpha}.L_{\alpha-1}^{\gamma}$ 
    & $D_{\beta}.L_{\beta}^{\gamma}$   
    & $M^{dda}$ & $agg(M^0)$ \\

$q^{dd^{\beta}}$  
    & $C^{0}$ 
    & $\phi_{\alpha}$                     
    & $\phi_{\beta}$                    
    & $\phi_{\Box}$                
    & $D_{\alpha}.L_{\alpha}^{\gamma}$   
    & $D_{\beta}.L_{\beta-1}^{\gamma}$ 
    & $M^{ddb}$ & $agg(M^0)$\\

$q^{s^{\alpha}}$  
    & $C^{0}$ 
    & $\phi_{\alpha}^{\star}$                 
    & $\phi_{\beta}$             
    & $\phi_{\Box}$                
    & $D_{\alpha}.L_{\alpha}^{\sigma}$  
    & $D_{\beta}.L_{\beta}^{\gamma}$   
    & $M^{ssa}$ & $agg(M^0)$ \\

$q^{s^{\beta}}$  
    & $C^{0}$ 
    & $\phi_{\alpha}$                 
    & $\phi_{\beta}^{\star}$             
    & $\phi_{\Box}$                
    & $D_{\alpha}.L_{\alpha}^{\gamma}$  
    & $D_{\beta}.L_{\beta}^{\sigma}$   
    & $M^{ssb}$ & $agg(M^0)$
\end{tabular}
\caption{Table for the proof of Multi-Query Usability of the all-encompassing query. The numbers in the header correspond to the id of the respective requirement in the Cube Usability Theorem.}
\label{tab:mqoProof1}
\end{table*}

\begin{table*}[h!]
    \centering
    \begin{tabular}{lllll}

Query & $D_\alpha$: \textit{Date} & $D_\beta$: \textit{Customer} 
    & $Grouper_\alpha$ & $Grouper_\beta$ \\
    \hline

$q^{A}$ 
&{\color{blue}{$Year~=~'1997'$}}
&{\color{blue}{$Country~=~'USA'$}}
&[{\color{blue}{$Day$}}
&{\color{blue}{$Customer$}}]\\
\hline
    
$q^{org}$ & $Quarter~=~1997Q3$ & $State~=~'CA'$ & $[Month$ & $Region]$\\

$q^{dd^{\color{blue}{Date}}}$ 
&$Quarter~=~1997Q3$ 
&$State~=~'CA'$ 
&[{\color{blue}{$Day$}}
&$Region$] \\

$q^{dd^{\color{blue}{Cust}}}$ 
&$Quarter~=~1997Q3$ 
&$State~=~'CA'$ 
&[$Month$
&{\color{blue}{$Customer$}}]\\

$q^{s^{\color{blue}{Date}}}$ & {\color{blue}{$Year~=~'1997'$}} & $State~=~'CA'$ 
& [{\color{blue}{$Quarter$}}  &$Region$]\\

$q^{s^{\color{blue}{Cust}}}$ &$Quarter~=~1997Q3$ 
&{\color{blue}{$Country~=~'USA'$}}
&[$Month$ &{\color{blue}{$State$}}]\\

    \end{tabular}
    \caption{Reference example revisited for the sake of theoretical explanations}
    \label{tab:motExMQOTheory}
\end{table*}

Let us summarize the situation concerning the selection atom of dimension $D_{\alpha}$. 
\begin{itemize}
    \item The ``previous'' query, which is the all-encompassing one, $q^A$, has a grouping level $D_{\alpha}.L_{\alpha-1}^{\gamma}$ and a selection condition {$\phi_{\alpha}^{\star}$: $L_{\alpha+1}^{\sigma}$ = $anc_{L_\alpha^\sigma}^{L_{\alpha+1}^\sigma}(v_\alpha)$}. So practically, the produced grouper domain of $q^A$ values for $D_{\alpha}$  are:\\
    $v^g$: s.t., 
    $anc_{L_{\alpha-1}^\gamma}^{L_{\alpha+1}^\sigma}(v^g)$ = $anc_{L_\alpha^\sigma}^{L_{\alpha+1}^\sigma}(v_\alpha)$,\\
    i.e., the children of $v_\alpha$'s mother ($v_\alpha$'s siblings) at the level $D_{\alpha}.L_{\alpha-1}^{\gamma}$.
    
    \item Concerning $D_{\alpha}$, all queries have either an atom $\phi_{\alpha}$: $L_{\alpha}^{\sigma}$ = $v_\alpha$, or an atom $\phi_{\alpha}^{\star}$: $L_{\alpha+1}^{\sigma}$ = $anc_{L_\alpha^\sigma}^{L_{\alpha+1}^\sigma}(v_\alpha)$.

    \item For the ``new'' queries with an atom $\phi_{\alpha}^{\star}$, there is nothing to be done: the all-encompassing one query has restricted the domain to exactly the needed values. 
    

    \item The rest of the queries, must adapt the list of returned values from $q^A$ to the ones they actually want, which are the ones satisfying $\phi_\alpha$. If this is the case, then an appropriate selection condition has to be applied to the results of $q^A$ to select only the necessary values.\\
    The transformation of the $\phi_\alpha$ atom to the level of $q^A$ schema, $D_{\alpha}.L_{\alpha-1}^{\gamma}$ is:\\  $D_{\alpha}.L_{\alpha-1}^{\gamma}$ $\in$ $desc_{L_{\alpha}^{\sigma}}^{L_{\alpha-1}^{\gamma}}(v_\alpha)$ \\
    which produces the children of $v_\alpha$, and, which, in turn is a clear subset of the grouper domain of $q^A$. So, if we apply this selection predicate to the results of $q^A$ we restrict the set of $D_{\alpha}$ values to exactly the ones we need.
\end{itemize}

The same process is applied to $\phi_\beta$ and $\phi_\beta^{\star}$.\\

Then, the (vi) item of the checklist is proved, and, therefore, the entire theorem holds.
\end{proof}

Getting back to our example, Table~\ref{tab:motExMQOTheory} is showing only the parts of relevance from the query definitions. Observe that, concerning the $Date$ dimension, the values produced by $q^A$ are the days (grouper level is $Day$) of 1997 (selection atom is $Year~=~1997$), i.e., $1997/01/01, 1997/01/02$, etc. Now, for all the queries who want the 3rd quarter of 1997, we need to reapply the filter for $Quarter~=~1997Q3$, \textit{but at the grouper level of $q^A$, i.e., Day}! Therefore, the filter applied to the results of $q^A$ will be $Day$ $\in$ $desc_{Quarter}^{Day}(1997Q3)$. The fact that we filter at a high level and we group at a lower level guarantees us perfect rollability \cite{DBLP:conf/dolap/Vassiliadis23, PV22}, which in turn allows us to reuse the results of the all encompassing query. Observe also that due to the fact that the groupers of $q^A$ are lower or equal to the ones of all the other queries, we can further group the values of $q^A$ to the higher levels of all the rest of the \ANALYZE internal queries.

\begin{algorithm}
\caption{Algorithm MaxMQO}\label{algo:AnalyzeMQO}
\SetAlgoLined
\KwIn{Original query $q^{org}$ = $\langle$ $C^{0}$, $\phi$, [$D_{\alpha}.L_{\alpha}^{\gamma}$, $D_{\beta}.L_{\beta}^{\gamma}$,$M$],
$agg(M^0)$ $\rangle$ , 
$phi$: $\phi_{\alpha}$ $\wedge$ $\phi_{\beta}$ $\wedge$ $\phi_{\Box}$, $phi_i:$ $D_{i}.L_{i}^{\sigma}$ = $v_{i}$, $i: \alpha~or~\beta$
}
\KwOut{Updated results for $q^{org}$, $q^{s^{A}}$, $q^{s^{B}}$, $q^{dd^{A}}$, $q^{dd^{B}}$
}

\Begin{
    Let $H^{org},H^{s^A},H^{s^B},H^{dd^A},H^{dd^B}$ be maps of the form $H: \langle~dom(L_1^{\gamma}),dom(L_2^{\gamma})~\rangle \rightarrow dom(M^0)$ \;
    Let $aggF^{\star} = adapter(aggF)$\;
    
\hspace*{\fill} 
    
    \tcc*[h]{Construct and execute the MQO query} 
    
    Let $\phi_{\alpha}^{\star}$: $L_{\alpha+1}^{\sigma}$ = $anc_{L_\alpha^\sigma}^{L_{\alpha+1}^\sigma}(v_\alpha)$ and
    $\phi_{\beta}^{\star}$: $L_{\beta+1}^{\sigma}$ = $anc_{L_\beta^\sigma}^{L_{\beta+1}^\sigma}(v_\beta)$
    \;
    Let $q^{A}$ = $\langle$ $C^{0}$, $\phi_{\alpha}^{\star}~\wedge~\phi_{\beta}^{\star}~\wedge~\phi_{\Box}$, 
    [$L_{\alpha-1}^{\gamma}, L_{\beta-1}^{\gamma}, L_{\alpha}^{\gamma}, L_{\beta}^{\gamma}, L_{\alpha}^{\sigma}, L_{\beta}^{\sigma}$,
    $M^A$], $agg(M^0)$ $\rangle$
    \;
    $q^{A}.cells$ = $q^{A}.execute()$\;

\hspace*{\fill} 
    
    \ForEach{tuple $t$ in the result set $q^{A}.cells$}{
        Assume $t$ = $\langle v_{\alpha-1}^{\gamma}, v_{\beta-1}^{\gamma}, v_{\alpha}^{\gamma}, v_{\beta}^{\gamma}, v_{\alpha}^{\sigma}, v_{\beta}^{\sigma}, m \rangle$\;
        
        \tcc*[h]{if t satisfies $\phi_\alpha \wedge \phi_\beta$: update org \& dd's}
        
        \If{$\sigma_{\phi_\alpha \wedge \phi_\beta}(t)$}{ 
            $updateMap(H^{org}, \langle \langle v_{\alpha}^{\gamma}, v_{\beta}^{\gamma} \rangle, m\rangle,agg)$\;
            $updateMap(H^{dd^{A}}, \langle \langle v_{\alpha-1}^{\gamma}, v_{\beta}^{\gamma} \rangle, m\rangle,agg)$\;
            $updateMap(H^{dd^{B}}, \langle \langle v_{\alpha}^{\gamma}, v_{\beta-1}^{\gamma} \rangle, m\rangle,agg)$\;            
        }
        \If{$\sigma_{\phi_\beta}(t)$}{
            $updateMap(H^{s^A}, \langle \langle v_{\alpha}^{\sigma}, v_{\beta}^{\gamma} \rangle, m\rangle,agg)$\;
        }
        \If{$\sigma_{\phi_\alpha}(t)$}{    
            $updateMap(H^{s^B}, \langle \langle v_{\alpha}^{\gamma}, v_{\beta}^{\sigma} \rangle, m\rangle,agg)$\;
        }
    }

    $q^{org}.cells = H^{org}$; $q^{dd^{A}}.cells = H^{dd^{A}}, q^{dd^{B}}.cells = H^{dd^{B}}$\;
    $q^{s^{A}}.cells = H^{s^{A}}; q^{s^{B}}.cells = H^{s^{B}}$ \;
    
    \Return $\langle q^{org}.cells, q^{s^{A}}.cells, q^{s^{B}}.cells, q^{dd^{A}}.cells, q^{dd^{B}}.cells \rangle$\;
}

\hspace*{\fill} 

\SetKwFunction{FupdateMap}{updateMap}
\SetKwProg{Fn}{Function}{:}{}
\Fn{\FupdateMap{$H:map$, $\langle \langle g_1,g_2\rangle, m \rangle$: $Tuple\langle Pair \langle grouper,grouper \rangle,measure \rangle$, $aggF^{\star}: agg.~Func.$}}{

    \If{$H.containsKey(\langle g_1,g_2\rangle)$}{ 
        $curValue$ = $H.get(\langle g_1,g_2\rangle)$\;
        $H.put(\langle g_1,g_2\rangle, aggF^{\star}(currValue,m))$\;
    }
    \Else{
    $H.put(\langle g_1,g_2\rangle, m)$\;
    }
}

\end{algorithm}

\subsection{The {Max Multi-Query Optimization} Algorithm}
The \textbf{MaxMQO Algorithm} (Algorithm \ref{algo:AnalyzeMQO}) describes how to compute all the facilitator queries of the ANALYZE operator with single access to the underlying database. The algorithm takes as input the definition of the original query, which is sufficient to determine the entire set of facilitator queries and to produce their results as output. The steps of the algorithm is described below.

\paragraph{Step 1.} First, the algorithm constructs five maps, one for each of the queries that determine the ANALYZE operator. Each of these maps will store the result of the respective query. The key of the map is the combination of grouper values (which are pretty much the coordinates of each result cell), and the value is the aggregate measure that pertains to these coordinates.
\paragraph{Step 2.} Second, the algorithm constructs and executes the all-encompassing auxiliary query $q^A$ that will serve as the basis to compute all the other query results.
\paragraph{Step 3.} Once the tuples of $q^A$ have been computed, we need to distribute them to the facilitator  queries of the operator. We visit each tuple at a time. All tuples will end up in the siblings, as the selection condition of $q^A$ is exactly that of the siblings. Recall that the selection condition of the siblings is broader than the one of the original query and the drill-downs (which is identical to the original one). However, some of the resulting tuples also pertain to the original and the drill-down queries, and so, they have to be pushed towards the respective maps. In the end, the hash-maps contain the results of the auxiliary queries.

\hrulefill

For example, assume that the user issues the following \textit{ANALYZE} query: \\

\noindent\textsf{ANALYZE sum(store\_sales) FROM sales}\\
\textsf{FOR \textcolor{blue}{state = 'CA'} AND \textcolor{blue}{quarter = '2025-Q4'} }\\
\textsf{GROUP BY state,quarter} 

This query is regarded as the original query $q^{org}$. Step 1 constructs five maps for each facilitator query result $H^{org}, H^{s^A}, H^{s^B}, H^{dd^A}, H^{dd^B}$. In step 2, the all-encompassing query $q^A$ is constructed by replacing the selection condition of $q^{org}$ with their parent levels and values. In addition, four more groupers are added in the current set, which contains the child and parent levels of the $q^{org}$ groupers. The all-encompassing query is: \\

\noindent\textsf{ANALYZE sum(store\_sales) FROM sales}\\ 
\textsf{FOR \textcolor{blue}{country = 'USA'} AND \textcolor{blue}{year = '2025'} }\\
\textsf{GROUP BY \textcolor{red}{city, month}, state, quarter, \textcolor{red}{country, year}}\\

Then, this all-encompassing query is translated to SQL and executed. The result contains the result tuples of all facilitator queries. In Step 3, each result tuple is visited and distributed to a result bucket $H^{org}, H^{s^A}, H^{s^B}, H^{dd^A}, H^{dd^B}$ that was created in Step 1 if it is part of the result of the respective facilitator query ($q^{org}, q^{s^A}, q^{s^B}, q^{dd^A}, q^{dd^B}$), with respect to the cube usability theorem \cite{PV22,DBLP:conf/dolap/Vassiliadis23}.

\hrulefill

Observe also how the tuples are integrated in any of the maps. The method \textit{updateMap()} implements this handling: (a) If the coordinates of the tuple do not already exist in the map, we have to insert the incoming tuple in the map, as the first of its group. (b) Otherwise, if the coordinates already exist, we need to update the aggregate measure. Bear in mind that we cannot directly apply the \textit{agg} function: whereas the aggregate function is directly applicable for the case of $min$, $max$ and $sum$, for the case of $count$, we simply have to increase the counter by one. To facilitate this, we introduce the \textit{adapter} aggregate function that takes care of this problem by adapting the actual combination of the current and the new value accordingly.

\subsection{Mid - Multi Query Optimization Algorithm for the Analyze Operator}
The Max Multi-Query Optimization Algorithm utilizes a
single query on the underlying database to retrieve the result of an
\emph{ANALYZE} query. However, the selection conditions of the all-encompassing
auxiliary query \emph{q\textsuperscript{A}} explore a vast super-set of
tuples that contain the query result for each of the five queries, plus
tuples that are not part of any query result, thus they are not aggregated. The spare tuples add performance latency, as they have to be processed by Max {MQO} to determine whether a tuple belongs to one of the five query result maps.

To reduce the vastness of the explored data space, we can reduce the
number of spare tuples returned by \emph{q\textsuperscript{A}}. Observe
that the selection conditions of the original query and the drill-down
queries are the same. We take advantage of that property to construct a
single query \emph{q\textsuperscript{MID}} that combines the original and drill-down facilitator queries, without touching the siblings. 

The \textbf{Mid {MQO} Algorithm} (Algorithm \ref{algo:Mid-MQO}) launches the two sibling queries along with the merged $q^{MID}$ query. Since the latter pertains to three facilitator queries, we
group the tuples on their original and ancestor levels to distribute the
result tuples. In that way, we get the tuples that
contribute to the original and drill-down queries result in a single
access. The rest of the processing is the same as with Max {MQO}. Compared to \emph{q\textsuperscript{A}}, \emph{q\textsuperscript{MID}} returns less detailed tuples and applies less groupers to them, thus reducing the result size.

Algorithm 2 presents the implementation of the \emph{Mid Multi-Query
Optimization}. First, the algorithm constructs three maps, one for the
original query and one for each drill-down query that determine the
\emph{ANALYZE} operator. Each of these maps will store the result of its
respective query. The key of the map is the combination of grouper
values and the value is the aggregate measure that pertains to these
coordinates (i.e.,
\textless grouper1,grouper\textsubscript{2}\textgreater{}
-\textgreater{} m). The resulting tuples are pushed into the respective
buckets via \emph{updateMap()} as described in section 5.2.

Then, the sibling queries q\textsuperscript{sA}, q\textsuperscript{sB}
are constructed and executed separately to return the exact sibling
result tuples. In this way, we avoid the presence of spare tuples by
applying selection conditions that involve only parent values and result
in exploring a large part of the fact table. Finally, the result tuples
of the three queries are packaged to create the ANALYZE query result.

\hrulefill

For example, assume that the user issues the following \textit{ANALYZE} query: \\

\noindent\textsf{ANALYZE sum(store\_sales) FROM sales}\\
\textsf{FOR state = 'CA' AND quarter = '2025-Q4'}\\
\textsf{GROUP BY state,quarter} 

This query is regarded as the original query $q^{org}$. The Mid-MQO algorithm constructs three maps for each facilitator query result $H^{org}, H^{dd^A}, H^{dd^B}$. Then, the merged query $q^{MID}$ is constructed by maintaining the selection conditions of $q^{org}$ and by adding, two more groupers in the current set. The groupers of the merged query contain the child levels of the $q^{org}$ groupers. The merged query is: \\

\noindent\textsf{ANALYZE sum(store\_sales) FROM sales}\\ 
\textsf{FOR state = 'CA' AND quarter = '2025-Q4'}\\
\textsf{GROUP BY \textcolor{red}{city, month}, state, quarter}\\

Then, the sibling queries $q^{s^A}, q^{s^B}$ are constructed with respect to the syntax discussed in Section \ref{sec:analyze}. Then, the merged query $q^{MID}$ and the sibling queries $q^{s^A}, q^{s^B}$ are translated to SQL and executed. The result of sibling queries contains the exact sibling result tuples. The merged query contains the result tuples of the original and drill-down queries. Similarly, as described in the Max-MQO algorithm, each tuple is visited and distributed to a result bucket $H^{org}, H^{dd^A}, H^{dd^B}$ if it is part of the respective facilitator query ($q^{org},q^{dd^A},q^{dd^B}$) with respect to the cube usability theorem \cite{PV22,DBLP:conf/dolap/Vassiliadis23}.

\hrulefill

\begin{algorithm}
\caption{Algorithm Mid Multi-Query Optimization}\label{algo:Mid-MQO}
\SetAlgoLined
\KwIn{Original query $q^{org}$ = $\langle C^{0}, \phi, [D_{\alpha}.L_{\alpha}^{\gamma}, D_{\beta}.L_{\beta}^{\gamma}, M], agg(M^0) \rangle$, \\
\hspace{2em}$\phi = \phi_{\alpha} \wedge \phi_{\beta} \wedge \phi_{\Box}$, \quad $\phi_i: D_{i}.L_{i}^{\sigma} = v_i, \; i \in \{\alpha,\beta\}$}
\KwOut{Updated results for $q^{org}, q^{s^A}, q^{s^B}, q^{dd^A}, q^{dd^B}$}

\Begin{
    Let $H^{org}, H^{s^A}, H^{s^B}, H^{dd^A}, H^{dd^B}$ be maps of the form \\
    \hspace{2em}$H: \langle dom(L^{\gamma}_1), dom(L^{\gamma}_2) \rangle \rightarrow dom(M^0)$\;
    Let $aggF^{\star} = adapter(aggF)$\;

    \tcc*[h]{Construct and execute Mid MQO queries}\\
    Let $\phi_{\alpha}^{\ast}: L_{\alpha+1}^{\sigma} = anc_{L_{\alpha+1}^{\sigma} L_{\alpha}^{\sigma}}(v_{\alpha})$ \;
    Let $\phi_{\beta}^{\ast}: L_{\beta+1}^{\sigma} = anc_{L_{\beta+1}^{\sigma} L_{\beta}^{\sigma}}(v_{\beta})$ \;

    Let $q^{org\&dd} = \langle C^0, \phi_{\alpha} \wedge \phi_{\beta}, [L_{\alpha-1}^{\gamma}, L_{\beta-1}^{\gamma}, L_{\alpha}^{\gamma}, L_{\beta}^{\gamma}, M^{org\&dd}], agg(M^0) \rangle$\;
    Let $q^{s^A} = \langle C^0, \phi_{\alpha}^{\ast} \wedge \phi_{\beta}, [L_{\alpha+1}^{\gamma}, L_{\beta}^{\gamma}, M^{s^A}], agg(M^0) \rangle$\;
    Let $q^{s^B} = \langle C^0, \phi_{\alpha} \wedge \phi_{\beta}^{\ast}, [L_{\alpha}^{\gamma}, L_{\beta+1}^{\gamma}, M^{s^B}], agg(M^0) \rangle$\;

    $q^{org\&dd}.cells \gets q^{org\&dd}.execute()$\;
    $q^{s^A}.cells \gets q^{s^A}.execute()$\;
    $q^{s^B}.cells \gets q^{s^B}.execute()$\;

    \ForEach{tuple $t \in q^{org\&dd}.cells$}{
        Assume $t = \langle v_{\alpha-1}^{\gamma}, v_{\beta-1}^{\gamma}, v_{\alpha}^{\gamma}, v_{\beta}^{\gamma}, m \rangle$\;

        \If{$\sigma_{\phi_{\alpha} \wedge \phi_{\beta}}(t)$}{
            $updateMap(H^{org}, \langle \langle v_{\alpha}^{\gamma}, v_{\beta}^{\gamma} \rangle, m \rangle, agg)$\;
            $updateMap(H^{dd^A}, \langle \langle v_{\alpha-1}^{\gamma}, v_{\beta}^{\gamma} \rangle, m \rangle, agg)$\;
            $updateMap(H^{dd^B}, \langle \langle v_{\alpha}^{\gamma}, v_{\beta-1}^{\gamma} \rangle, m \rangle, agg)$\;
        }
    }

    $q^{org}.cells \gets H^{org}$\;
    $q^{dd^A}.cells \gets H^{dd^A}$, $q^{dd^B}.cells \gets H^{dd^B}$\;
    \Return $\langle q^{org}.cells, q^{s^A}.cells, q^{s^B}.cells, q^{dd^A}.cells, q^{dd^B}.cells \rangle$\;
}

\SetKwFunction{FupdateMap}{updateMap}
\SetKwProg{Fn}{Function}{:}{}\Fn{\FupdateMap{$H:map, \langle \langle g_1,g_2 \rangle, m \rangle:$\\$ Tuple \langle Pair \langle grouper,grouper \rangle, measure \rangle, aggF^{\star}$}}{
    \If{$H.containsKey(\langle g_1,g_2 \rangle)$}{
        $curValue \gets H.get(\langle g_1,g_2 \rangle)$\;
        $H.put(\langle g_1,g_2 \rangle, aggF^{\star}(curValue, m))$\;
    }
    \Else{
        $H.put(\langle g_1,g_2 \rangle, m)$\;
    }
}
\end{algorithm}
\forcepagebreak

\subsection{A note on the design choices of our algorithms}
Before proceeding to present the experimental evaluation of these algorithms, it is worthwhile to stop for a quick comment on the fundamental design choices that guide them. We introduce an operator that ultimately invokes several database queries (even if we produce an algorithm to merge them into one). Also, the optimizations that we introduce are (a) DBMS-external, and therefore, DBMS-agnostic, and, (b) performed automatically at the syntactic level of query rewriting, rather than exploring alternative multi-query optimizations. 

\result{One important lesson to be learned here is that we can introduce a level of abstraction above the direct database querying that resembles more naturally the thinking of the analyst} (much like Roll-Ups and Drill-Downs were introduced for OLAP). These high level operators, internally, can involve several facilitator queries, and therefore incur execution costs, and, at the same time, provide optimization opportunities. 

At the same time, observe that all our algorithms \result{are using \emph{cube} queries, rather than \emph{database} queries, as the granule of work}. As a result, the rewritings of the original query that these algorithm perform is performed by exploiting the combination of filters and groupers at the hierarchical multidimensional space and do \textit{not} involve the exploration of alternative plans inside the DBMS. 

Overall, the entire design offers another important lesson to be learned: \result{working at the cube query level, \emph{outside} the DBMS can incur significant performance gains}.

Of course, integrating such concepts inside the DBMS is also worth exploring, although outside the scope of the current paper. If no optimizations were introduced, either inside or outside the DBMS, the execution of ANALYZE ends up being the execution of the Min-MQO algorithm, i.e, the batch execution of five database queries, one after the other. Apart from optimizing this outside the DBMS, we can also consider the possibility of expanding the DBMS with the operator, in which case, the optimizer of the DBMS could also be expanded too. This involves both the MQO possibilities listed here, and also, other opportunities --for example, deciding the join order in this context is worth exploring. 
Observe that, in contrast to traditional MQO algorithms, exactly because of the characteristics of the data space, we \textit{do not seek} to find rewritings, we automatically (and therefore efficiently) perform them, by exploiting the hierarchies.

\section{Experiments}\label{sec:exps}
In this section, we experimentally evaluate the execution time of the proposed algorithms under varying configurations of data and query workloads. In all of our experiments, we use the Delian Cubes system \cite{Delian_Cubes_Engine}, which is a cube query answering engine, within which we have incorporated our algorithms. All algorithms have been implemented in Java, as part of the Delian Cubes system.  We have used MySQL 8.0.34 as the underlying database. For our evaluation, the methods run on a Windows 11 2.50 GHz 14-core processor system with 32GB main memory and 1TB SSD. All material is found at: \url{https://github.com/DAINTINESS-Group/DelianCubeEngine}.

\subsection{Experimental Setup}
\paragraph{Experimental Goal and Evaluation Metrics.}
The main goal of our evaluation is to evaluate the efficiency of the 
proposed algorithms to execute the ANALYZE operator. Therefore, 
for all antagonists, we measure the \textit{Total Execution Time} needed to fully compute the answer to an ANALYZE query, as well as its breakdown to different facilitators 
within each algorithm. \textit{Total Query Execution Time} is the sum of the time taken by different parts of each algorithm, and specifically: (i) \textit{Parsing Time} for the string input expression of the ANALYZE query, (ii) \textit{Construction Time} for the facilitator queries, (iii) \textit{Facilitator Execution Time} for executing the facilitator queries, (iv) \textit{Post-processing Time} of the results of the facilitator queries and population of the ANALYZE query result correctly.

\paragraph{Competitor Algorithms.}
In our experiments, we compare the three different algorithms to 
implement the \textit{ANALYZE} operator on \textit{Total Execution 
Time} and its breakdown. For the convenience of the reader, we briefly 
summarize these three methods in both Table \ref{tab:T3} and 
the following short summary. All algorithms receive as input an ANALYZE 
query with two groupers, a distributive aggregate function, and at least 
two selection conditions. Recall that the semantics of the operator 
entail the execution of 5 internal \textit{facilitator} cube queries, 
namely the original one $q^{org}$, two sibling $q^{s^A}, q^{s^B}$ and two drilling cube queries $q^{dd^A}, q^{dd^B}$.
\begin{itemize}
    \item \textit{Min-MQO} is the algorithm that performs the least merging of the operator's underlying cube queries (thus, "Minimum Multi-Query Optimization"). As already mentioned, the algorithm constructs the five facilitator queries as described by the operator's semantics. Practically, this is the plain simple execution of the operator without any optimization and as such, it also avoids having to perform any post-processing to the results obtained from the executed queries. 
    \item \textit{Max-MQO} is the algorithm that performs the maximum merging of the facilitator queries into a single facilitator query (thus, "Maximum Multi-Query Optimization") with four extra groupers to cover the \textit{siblings} and \textit{drill-down} results, substituted selection conditions (input atoms are substituted by their parent levels and parent values). The results are post-processed and the tuples are distributed to the correct placeholder with respect to the operator semantics. 
    \item \textit{Mid-MQO} is the algorithm that strikes a balance in the middle of the two extremes of full -and no - merging of the facilitator queries into one. Specifically, the algorithm constructs three facilitator queries: (i) two \textit{sibling} queries, one for each grouping dimension, and, (ii) a merged \textit{original-n-drill-down} multi-query which is a combination of the \textit{original} and \textit{drill-down} queries, which exploits the fact that these queries have the exact same selection conditions and thus search the same data space of tuples in the fact table. The resulting tuples are post-processed and distributed to the correct placeholder as in the previous case. 
\end{itemize}

\begin{table}[tbh]
\centering
\begin{tabular}{r|r|r}
\textbf{Name} & \textbf{\# Executable Queries} & \textbf{Post-Processing} \\
\hline
Min-MQO & 5 & No \\
Mid-MQO & 3 & Yes \\
Max-MQO & 1 & Yes \\
\end{tabular}
\caption{Details of Methods}
\label{tab:T3}
\end{table}

\paragraph{Datasets.}
We have used the datasets listed in Table \ref{tab:T4} to perform our evaluation. Specifically, these datasets are: 
(i) \textit{Northwind} \cite{Northwind}, which is a synthetic dataset that represents the sales and the operations of a food import/export company, 
(ii) \textit{Foodmart}, which is a cubefied version of \textit{Foodmart} \cite{Foodmart}, which contains synthetic data regarding the sales activity of a retail supermarket, 
(iii) \textit{pkdd99+}, which is an upscaled version of \textit{pkdd99} \cite{DBLP:conf/pkdd/ZhongYO99}, a financial dataset that contains data about loans and transactions, for which we have created a fact table that contains 100 million entries, and, 
(iv) \textit{TPC-DS} \cite{TPC-DS}, which is an industry standard benchmark designed to evaluate the performance of analytical platforms. 
We have tested the scalability of our algorithms over three versions of \textit{TPC-DS} with scale factor 1, 3.5, and, 35. As Delian Cubes operates over data cubes, we do not directly utilize the raw schema of the datasets, but rather adapt it to a cubefied, clean, star-schema version, in order to support hierarchies and cube queries.

\begin{table}[tbh]
\centering
\begin{tabular}{r|r|r|r}

\textbf{Name} & \textbf{Scale Factor} & \textbf{\# Dim. Tables
} & \textbf{\# Facts} \\
\hline
Northwind & 1 & 4 & 2K \\

Foodmart & 1 & 5 & 288K \\

pkdd99+ & 166000 & 3 & 100M \\

TPC-DS & 1,3.5,35 & 6 & 2M,10M,100M \\

\end{tabular}
\caption{Dataset Basic Characteristics}
\label{tab:T4}
\end{table}

\begin{table}[p]
	\centering
	\begin{tabular}{c|c|c|c|c|c|c|c}
		
		\multirow{2}{*}{\textbf{QueryID}}& \multirow{2}{*}{\textbf{Filter Level}} & \multirow{2}{*}{\textbf{Grouper Level}}  & 	\multirow{2}{*}{\textbf{\#Filters}} & \multicolumn{4}{c}{\textbf{Selectivity Ratio}}  \\
		& & & & $q^{org}$ & $q^{sib_1}$ & $q^{sib_2}$ & $q^A$ \\
		\hline
		1 & Mid-High & Mid-High  & 2 & 1\% & 8\% & 3\% & 20\% \\
		
		2 & Mid-High & Mid-High & 2  & 1\% & 4\% & 4\% & 9\% \\
		
		3 & Low-High & Low-High & 2  & 1.5\% & 8\% & 6\% & 20\% \\
		
		4 & Mid-High & Mid-High & 2 & 2\% & 4\% & 8\% & 20\% \\
		
		5 & Mid-High & Mid-High & 2 & 3\% & 8\% & 9\% & 20\% \\
		
		6 & Mid-High & Mid-High & 2 & 3.5\% & 8\% & 3\% & 20\% \\
		
		7 & High-High & High-High & 2  & 6\% & 30\% & 20\% & 100\% \\

		8 & High-High & High-High & 2  & 6.5\% & 30\% & 20\% & 100\% \\

		9 & High-High & High-High & 2  & 8\% & 37\%  & 20\% & 100\% \\
		
		10 & High-High & High-High & 2  & 8\% & 37\% & 20\% & 100\% \\
	\end{tabular}
	\caption{TPC-DS Time Workload Characteristics}
    \label{tab:T5}
\end{table}

\begin{table}[p]
	\centering
	\begin{tabular}{c|c|c|c|c|c|c|c}
		
		\multirow{2}{*}{\textbf{QueryID}}& \multirow{2}{*}{\textbf{Filter Level}} & \multirow{2}{*}{\textbf{Grouper Level}}  & 	\multirow{2}{*}{\textbf{\#Filters}} & \multicolumn{4}{c}{\textbf{Selectivity Ratio}}  \\
		& & & & $q^{org}$ & $q^{sib_1}$ & $q^{sib_2}$ & $q^A$ \\
		\hline
		1 & Low-Low & Low-Low  & 2 & 0.001\% & 0.001\% & 0.35\% & 1\% \\
		
		2 & Mid-Low & Mid-Low & 2  & 0.002\% & 0.01\% & 0.35\% & 2\% \\
		
		3 & High-Low & High-Low & 2  & 0.01\% & 0.03\% & 1\% & 10\% \\
		
		4 & Low-High & Low-High & 2 & 0.2\% & 0.5\% & 2\% & 5\% \\
		
		5 & Low-High & Low-High & 2 & 0.35\% & 0.5\% & 2\% & 9\% \\
		
		6 & Mid-High & Mid-High & 2 & 0.5\% & 1\% & 3.5\% & 20\% \\
		
		7 & Mid-High & Mid-High & 2  & 0.9\% & 2\% & 6\% & 20\% \\
		
		8 & Mid-High & Mid-High & 2  & 1\% & 10\% & 20\% & 20\% \\
		
		9 & High-High & High-High & 2 & 2\% & 2\%  & 9\% & 100\% \\
		
		10 & High-High & High-High & 2 & 2\% & 9\% & 10\% & 100\% \\
	\end{tabular}
	\caption{TPC-DS Item Workload Characteristics}
    \label{tab:T6}
\end{table}

\begin{table}[p]
	\centering
	\begin{tabular}{c|c|c|c|c|c|c|c}
		
		\multirow{2}{*}{\textbf{QueryID}}& \multirow{2}{*}{\textbf{Filter Level}} & \multirow{2}{*}{\textbf{Grouper Level}}  & 	\multirow{2}{*}{\textbf{\#Filters}} & \multicolumn{4}{c}{\textbf{Selectivity Ratio}}  \\
		& & & & $q^{org}$ & $q^{sib_1}$ & $q^{sib_2}$ & $q^A$ \\
		\hline
		1 & Low-Low & Low-Low  & 2 & 0.01\% & 0.01\% & 0.2\% & 0.2\% \\
		
		2 & Mid-Low & Mid-Low & 2  & 0.01\% & 0.25\% & 0.15\% & 3\% \\
		
		3 & Mid-Low & Mid-Low & 2  & 0.2\% & 0.2\% & 2\% & 2\% \\
		
		4 & High-Low & High-Low & 2 & 0.6\% & 3\% & 4\% & 20\% \\
		
		5 & High-High & High-High & 2 & 1.4\% & 17\% & 11\% & 100\% \\
		
		6 & High-High & High-High & 2 & 1.7\% & 17\% & 9\% & 100\% \\
		
		7 & High-Low & High-Low & 2  & 2\% & 2\% & 13\% & 13\% \\
		
		8 & Mid-High & Mid-High & 2  & 3\% & 17\% & 18\% & 100\% \\
		
		9 & High-High & High-High & 2 & 3\% & 17\%  & 18\% & 100\% \\
		
		10 & High-High & High-High & 2 & 3\% & 17\% & 18\% & 100\% \\
	\end{tabular}
	\caption{pkdd99+ Workload Characteristics}
    \label{tab:T7}
\end{table}

\begin{table}[h]
	\centering
	\begin{tabular}{c|c|c|c|c|c|c|c}

		\multirow{2}{*}{\textbf{QueryID}}& \multirow{2}{*}{\textbf{Filter Level}} & \multirow{2}{*}{\textbf{Grouper Level}}  & 	\multirow{2}{*}{\textbf{\#Filters}} & \multicolumn{4}{c}{\textbf{Selectivity Ratio}} \\
		& & & &   $q^{org}$ & $q^{sib_1}$ & $q^{sib_2}$ & $q^A$ \\
		\hline
		1 & Low-Low & Low-Low & 2 & 0.001\% & 1\% & 0.01\% & 3\% \\

		2 & Low-Low & Mid-Mid & 2 & 0.001\% & 1\% & 0.01\% & 3\% \\

		3 & Low-Low & High-High & 2 & 0.001\% & 1\% & 0.01\% & 3\% \\

		4 & Mid-Mid & Low-Low & 2 & 3\% & 9\% & 9\% & 30\% \\

		5 & Mid-Mid & Mid-Mid & 2 & 3\% & 9\% & 9\% & 30\% \\

		6 & Mid-Mid & High-High & 2 & 3\% & 9\% & 9\% & 30\% \\

		7 & High-Mid & Low-Low & 2 & 9\% & 30\% & 19\% & 67\% \\

		8 & High-High & Low-Low & 2 & 30\% & 30\% & 67\% & 100\% \\

		9 & High-High & Mid-Mid & 2 & 30\% & 30\% & 67\% & 100\% \\

		10 & High-High & High-High & 2 & 30\% & 30\% & 67\% & 100\% \\

	\end{tabular}
	\caption{Foodmart Workload Characteristics}
    \label{tab:T8}
\end{table}

\begin{table}[h]
	\centering
	\begin{tabular}{c|c|c|c|c|c|c|c}
		\multirow{2}{*}{\textbf{QueryID}}& \multirow{2}{*}{\textbf{Filter Level}} & \multirow{2}{*}{\textbf{Grouper Level}}  & 	\multirow{2}{*}{\textbf{\#Filters}} & \multicolumn{4}{c}{\textbf{Selectivity Ratio}}  \\
		& & & &   $q^{org}$ & $q^{sib_1}$ & $q^{sib_2}$ & $q^A$ \\
		\hline
		1 & Low-Low & Low-Low & 2 & 0.5\% & 6\% & 1\% & 13\% \\

		2 & Low-Low & Mid-Mid & 2 & 0.5\% & 6\% & 1\% & 13\% \\

		3 & Low-Low & High-High & 2 & 0.5\% & 6\% & 1\% & 13\% \\

		4 & Mid-Mid & Low-Low & 2 & 1\% & 13\% & 1\% & 17\% \\

		5 & Mid-Mid & Mid-Mid & 2 & 1\% & 13\% & 1\% & 17\% \\

		6 & High-Mid & Low-Low & 2 & 6\% & 28\% & 13\% & 74\% \\

		7 & High-Mid & High-High & 2 & 6\% & 28\% & 13\% & 74\% \\

		8 & High-High & Low-Low & 2 & 13\% & 74\% & 17\% & 100\% \\

		9 & High-High & Mid-Mid & 2 & 13\% & 74\% & 17\% & 100\% \\

		10 & High-High & High-High & 2 & 13\% & 74\% & 17\% & 100\% \\

	\end{tabular}
	\caption{Northwind Workload Characteristics}
    \label{tab:T9}
\end{table}

\paragraph{Query Workloads.}
The query workloads cover multiple cases of an ANALYZE query. The characteristics of an ANALYZE query are (i) the \textit{selectivity ratio}, (ii) the \textit{level of filter/grouper values}, and (iii) \textit{the number of filters}. The selectivity ratio refers to the percentage of tuples in the fact table that are involved in the computation of the final query result. The level of filter/grouper values is a relative characterization of the height of the filter/grouper level in the hierarchy of its dimension. 
Tables \ref{tab:T5}, \ref{tab:T6}, \ref{tab:T7}, \ref{tab:T8}, \ref{tab:T9} summarize the workload characteristics of each dataset. To assess the operator with different selectivities over the fact tables, for each data set, we use a workload of 10 queries, where we (a) vary with the various levels of filters/groupers, and (b) employ one large and one small dimension table. For the TPC-DS dataset, and in order to also evaluate the effect of dimension table size, we use two query workloads, both involving the large \textit{Date} dimension: (i) a workload using the large \textit{Time} dimension table and (b) a workload using the small \textit{Item} dimension table. In all workloads, as the QueryID increases, the selectivity of $q^{org}$ also increases in all data sets. 
The two workloads differ only on the selectivity ratio due to the difference in the dimension table size. The \textit{Time} workload involves two relatively large dimension tables and the \textit{Item} workload one large and one small. Both workloads involve the large \textit{Date} dimension. For the rest of the datasets, we employed a single workload of 10 queries with various levels of filters/groupers that involve one large and one small dimension table. To ensure that the result size of each ANALYZE query is relatively small in order to be comprehensible and manageable, we set the grouper levels of each ANALYZE query at the highest level possible with respect to the constraints discussed in Section \ref{sec:background}.

\subsection{Effect of Query Characteristics on the Performance of the 
Algorithms}
In this subsection, we discuss the effect of each of the query 
characteristics on the overall performance of the antagonist algorithms. 
In all the presented Figures, each bar represents the total execution 
time, with Min-MQO shown as a dark blue bar, Mid-MQO as a blue bar, Max-MQO 
as a light blue bar, and the timeout limit as a red dotted line. In the experiments that follow, we want to assess the effect of each of these characteristics on the performance of the different algorithms. To this end, we define a reference query and modify the values 
of the assessed characteristic while keeping the rest of the characteristics fixed. We used the \textit{TPC-DS} dataset that contains 100 million facts. The reference values for the query characteristics are: 20\% selectivity ratio, mid grouper levels, and 2 filters per query. The values of each tested query characteristic are prescribed in the respective subsection.

\subsubsection{Effect of Selectivity Ratio}
In this experiment, we evaluate how the selectivity ratio of the ANALYZE 
query affects its execution time for the three antagonist algorithms. 
Figure \ref{fig:F2} presents the effect of an ANALYZE query 
selectivity ratio on the query execution time for all methods. The 
selectivity ratios that we utilized for our experiment are 
10\%/50\%/90\% of the number of fact table tuples. For all antagonists, 
the query execution time scales, due to the increasing number of tuples 
processed. 

Observe that in all these setups, Mid-MQO consistently outperforms the antagonist 
methods. For low selectivity, Min-MQO and Max-MQO are close because the number of tuples being processed is small. Both Min-MQO and Mid-MQO scale similarly with selectivity; on the other hand, Max-MQO quickly deteriorates as selectivity increases (recall that internally, 
Max-MQO has the largest interim result set to process).

\silence{In all occasions, however, Mid-MQO wins; in particular, in the practically 
quite popular lower-selectivity case, its benefits are quite obvious.}

\begin{figure}[h]
    \centering
    \includegraphics[width=0.8\linewidth]{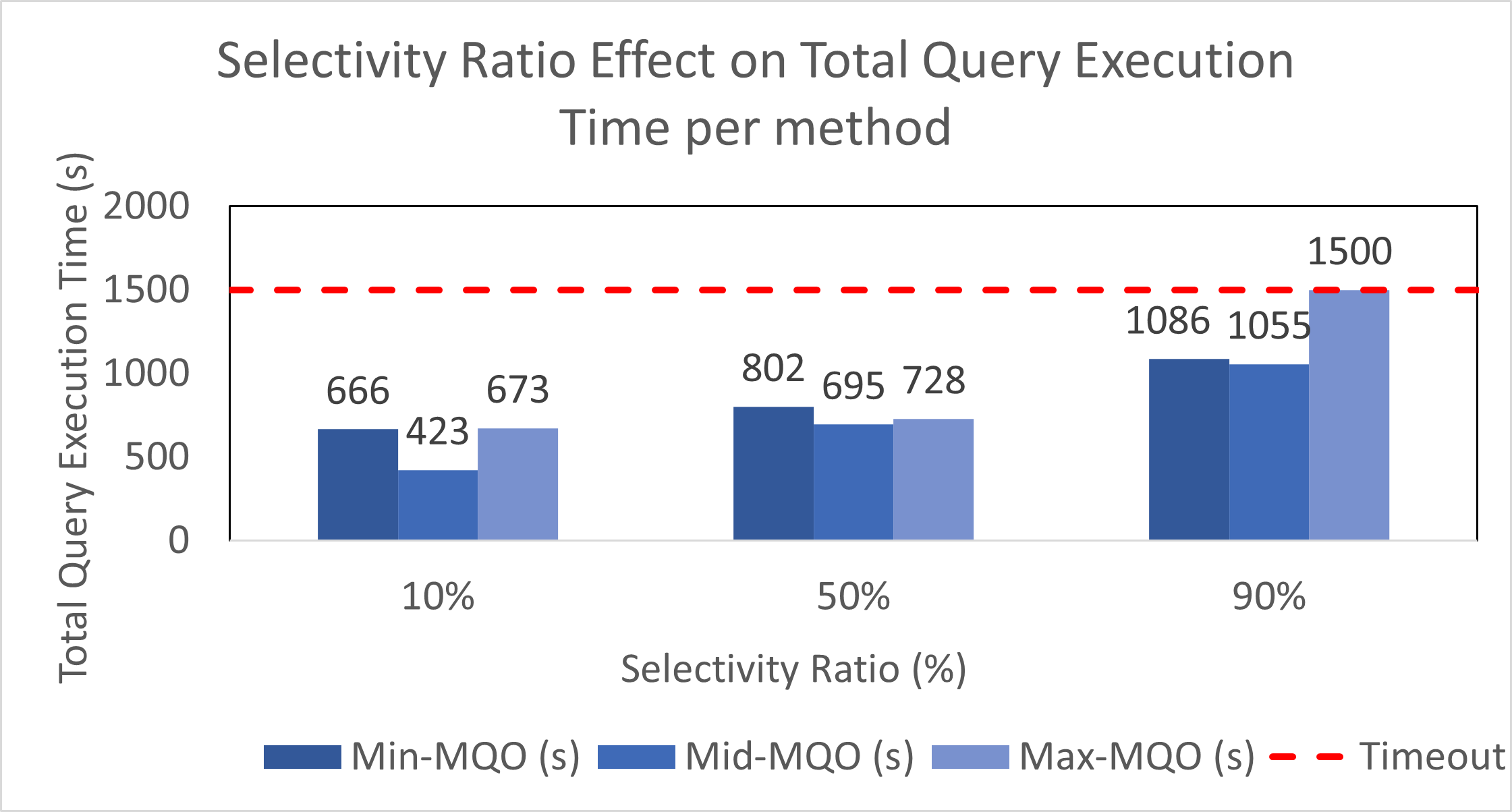}
    \caption{Selectivity Ratio effect on Total Querying Execution Time. \textit{Max-MQO} did not complete its execution on the 90\% selectivity ratio case.}
    \label{fig:F2}
\end{figure}

\silence{\begin{figure}[h]
    \centering
    \includegraphics[width=0.8\linewidth]{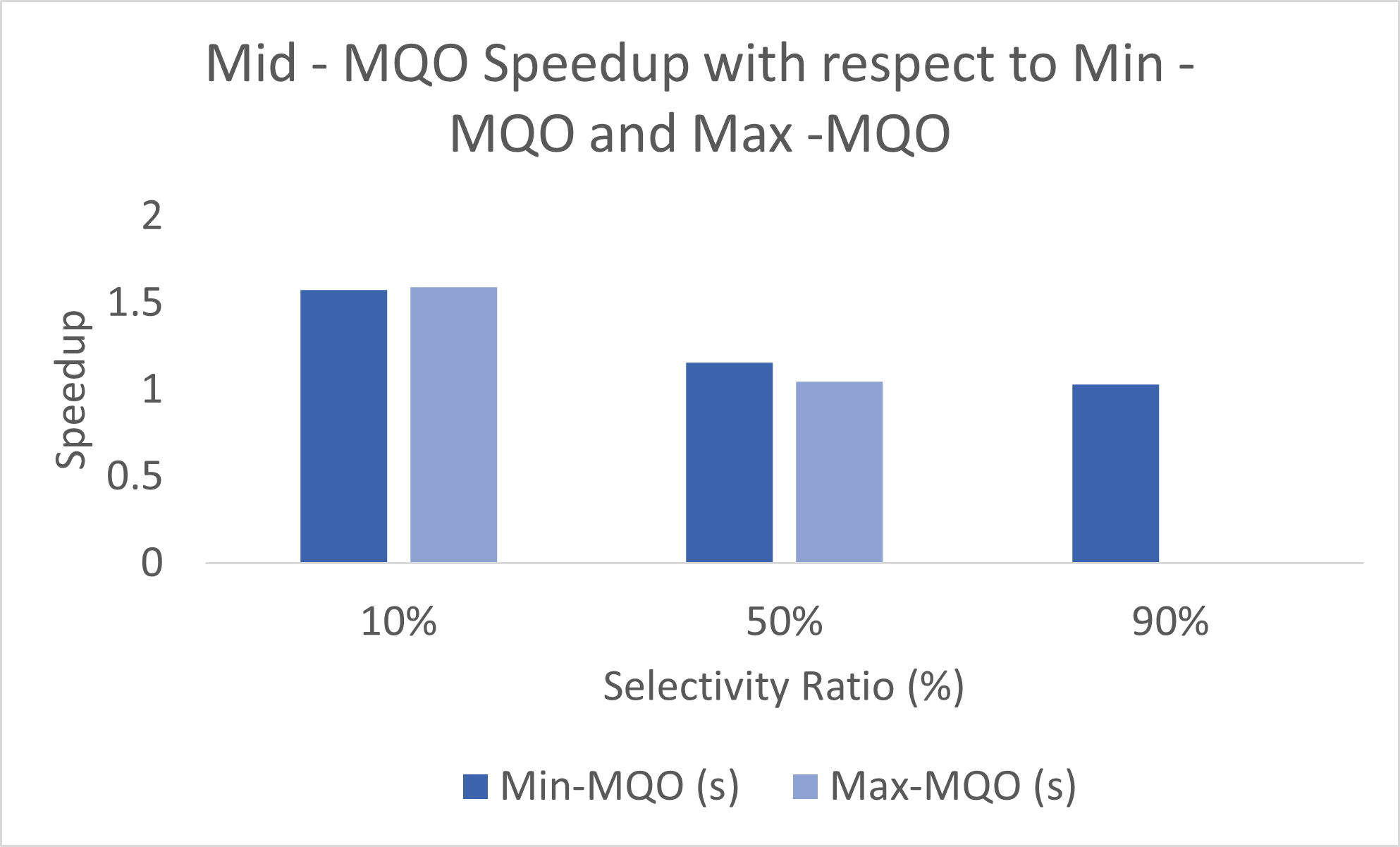}
    \caption{Mid - MQO Speedup with respect to antagonists. \textit{Max-MQO} did not complete its execution on the 90\% selectivity ratio case.}
    \label{fig:F3}
\end{figure}

Figure \ref{fig:F3} presents the speedup of Mid-MQO performance 
with respect to Min-MQO and Max-MQO performance. We define as \textit{speedup} the ratio $\frac{slow~algo}{fast~algo}$. In the low-selectivity 
case, Min-MQO outperforms its antagonists by approximately 1.5x. As the 
selectivity rises, speedup is decreasing to approximately 1.1x. In the 
high selectivity case, as depicted by Figure \ref{fig:F2},
Max-MQO execution reaches the timeout limit without completing the 
execution of the operator. Thus, we are unable to calculate the Min-MQO 
speedup for this case.}

\subsubsection{Effect of Number of Atomic Filters}
In this experiment, we evaluate how the number of atomic filters in 
the selection condition of an ANALYZE query affects its execution time 
for the three antagonist algorithms. Figure \ref{fig:F6}
shows the effect of the number of filters per query on the query 
execution time for all methods. The number of atomic filters utilized 
for our evaluation is 2/3/4/5. It is important to note that 
the presence of join operations and the number of filters are directly 
correlated: more filters require that more dimension tables be joined. 
Naturally, the number of filters also directly affects the selectivity ratio, as the number of atoms of the selection condition is increasing. This is due to the nature of the query and signifies that, although more joins are added, the effect of join processing is 
annulated by the number of tuples involved in the processing of the facilitator queries. As we observe Figure \ref{fig:F6}, a few facts become evident:

\begin{itemize}
\item As already mentioned, the decrease of the number of involved 
tuples directly shrinks the execution time for all algorithms.
\item Max-MQO gains performance much faster than the other algorithms as the data 
volume decreases due to the increase of the selectivity ratio, utilizing the single access to the database. Ultimately, in very low numbers of detailed fact tuples involved, Max-MQO wins.
\item Min-MQO and Mid-MQO are not affected by the shift in the number of filters in a  fashion similar to that of Max-MQO. Both of these methods are accessing the database multiple times and perform multiple join operations on multiple queries due to the extra filters. As the number of tuples is retained at high levels, this need for extra joins reduces the benefits of the tuples involved compared to Max-MQO, which accesses the database 
once. In all occasions but one, Min-MQO is the worst-performing algorithm.
\end{itemize}

\begin{figure}[h]
    \centering
    \includegraphics[width=0.75\linewidth]{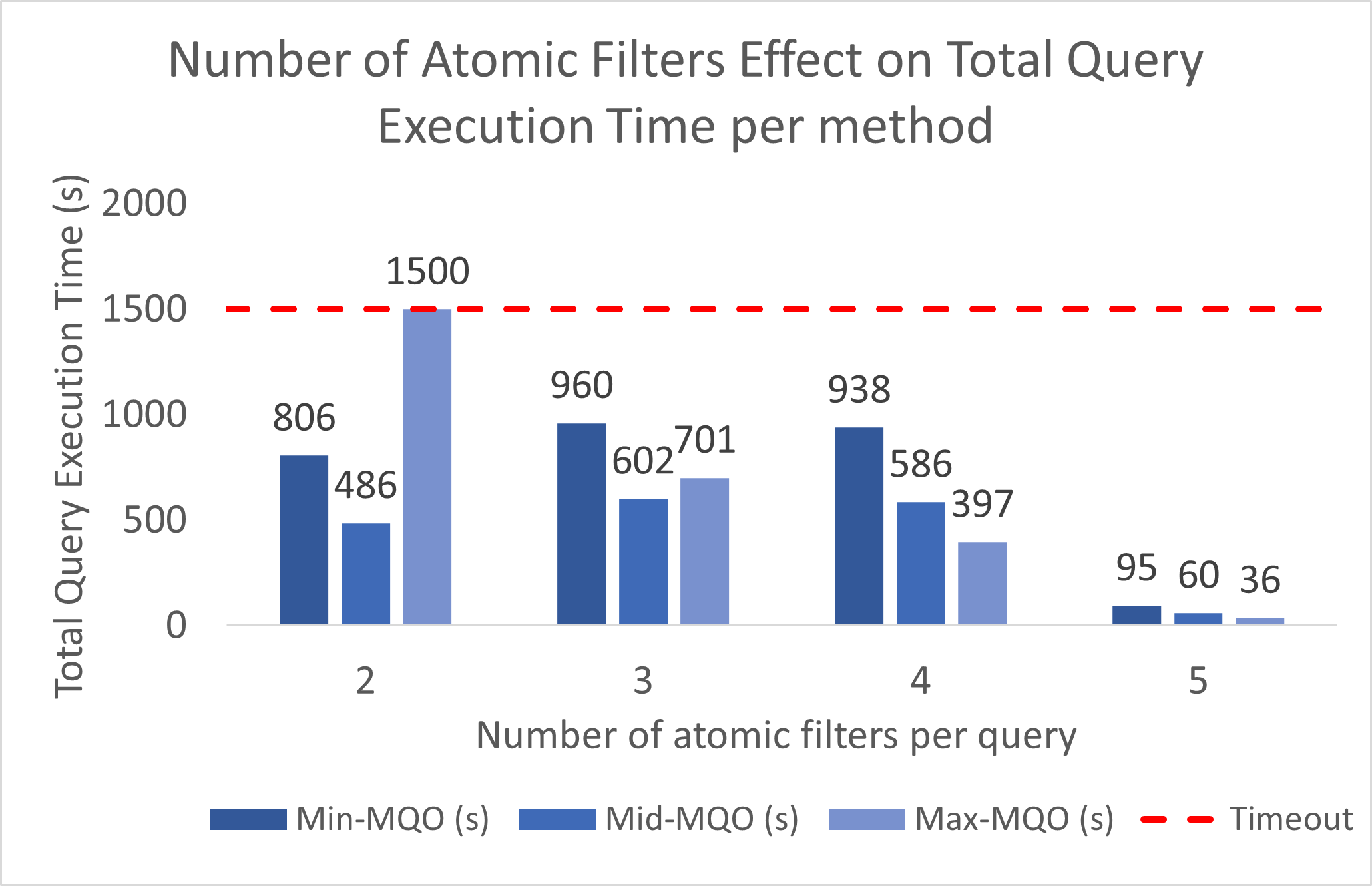}
    \caption{Number of Atomic Filters Effect. \textit{Max-MQO} did not complete its execution for 2 atomic filters.}
    \label{fig:F6}
\end{figure}

\begin{figure}[h]
    \centering
    \includegraphics[width=0.75\linewidth]{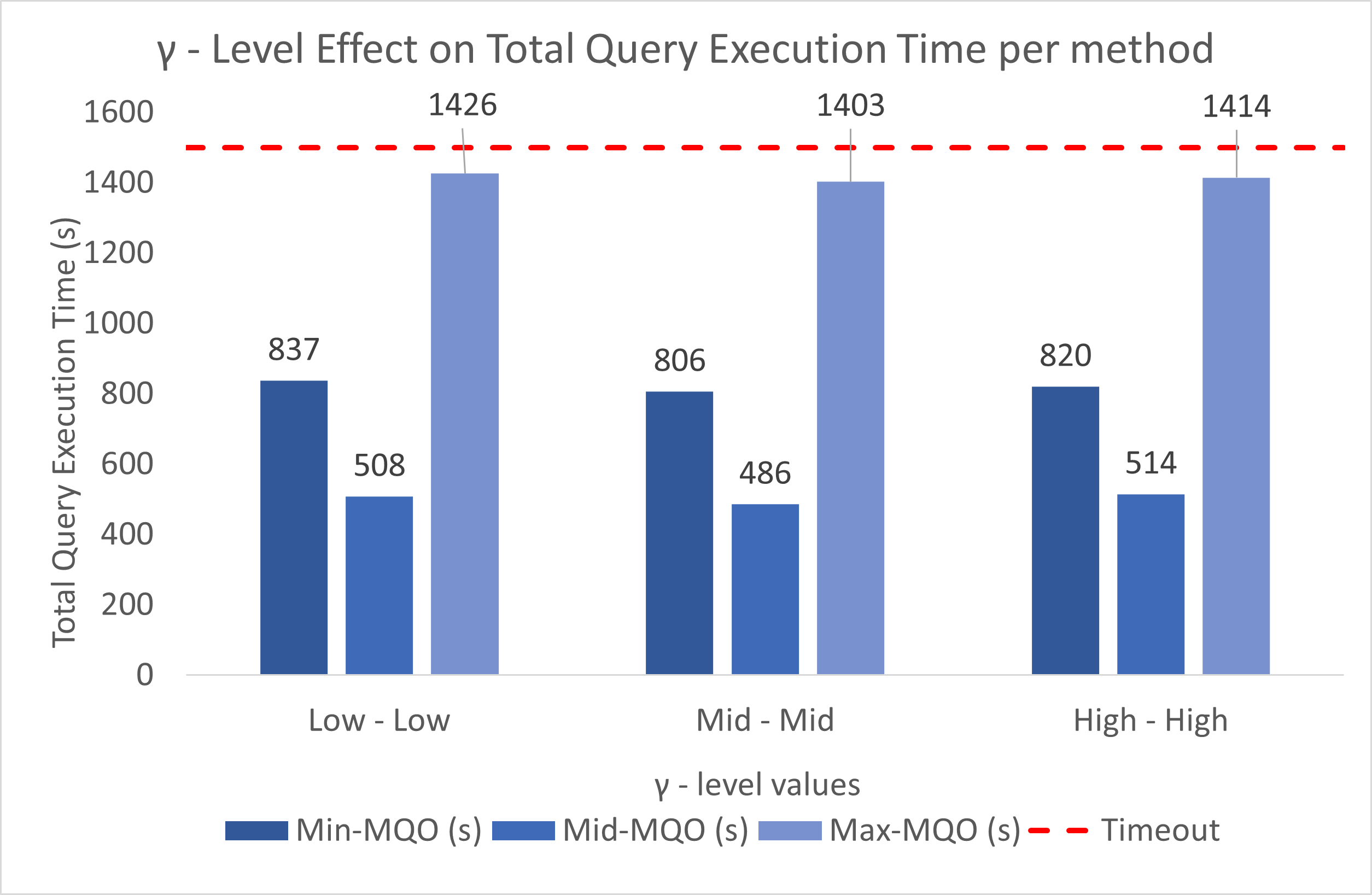}
    \caption{Grouper Effect on Total Execution Time}
    \label{fig:groueperEffect}
\end{figure}

\begin{figure}[h]
    \centering
    \includegraphics[width=0.75\linewidth]{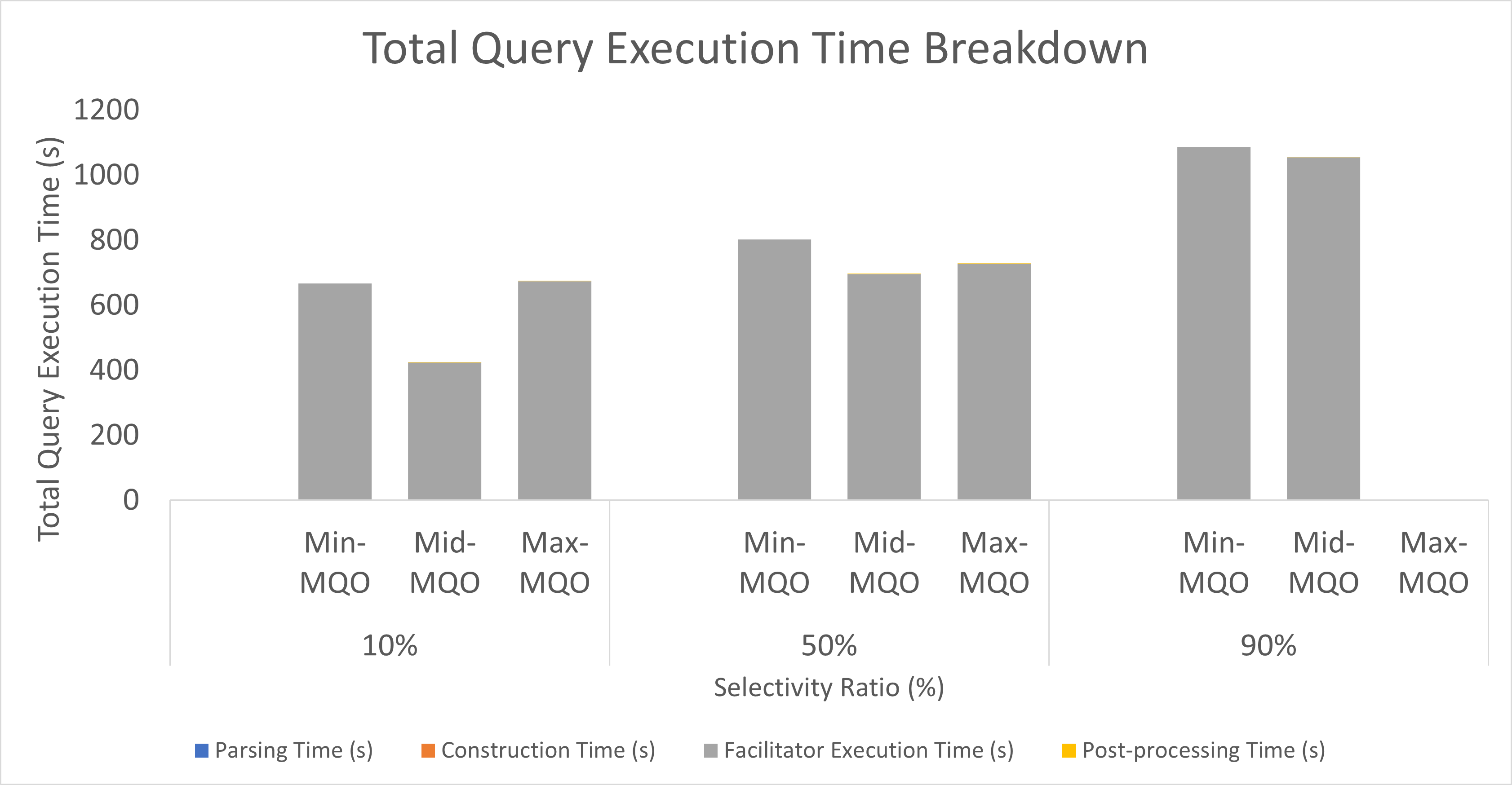}
    \caption{Total Execution Time Breakdown}
    \label{fig:executionBreakdown}
\end{figure}

\subsubsection{Other considerations} There are a couple of results that we briefly summarize here. 
\begin{itemize}
    \item \textit{Effect of the level of aggregation due to the groupers}. We have varied the ``height" of the grouper levels in their respective hierarchy. To this end, we classify the ``height" of a level as \textit{low} (when it is the lowest level of the hierarchy), \textit{high} (the topmost ALL, as well as its immediate child), and \textit{mid} otherwise. Our experiments have demonstrated that the level of aggregation of groupers has no effect on Total Execution Time (Figure \ref{fig:groueperEffect}). 
    \item \textit{Time breakdown}. We have analyzed the internal behavior of the alternative algorithms, to detect possible points of delay. To this end, we have split the query execution process to four stages, specifically (a) parsing, (b) construction of auxiliary facilitator queries, (c) execution time of all facilitator queries, and, (d) post-processing of the intermediate results, to construct the final query answer (no need for this in the case of Min-MQO). To cover the entire spectrum of fact table coverage by the original query, we varied the selectivity ratio of the original query to 10\%, 50\%, 90\% of the fact table tuples. The breakdown of Total Execution Time into the aforementioned four stages, reveals that in all cases, the Facilitator Execution Time takes approximately 99\% of Total Query Execution Time, making it the nexus of any optimization effort (Figure \ref{fig:executionBreakdown}). 
\end{itemize}

\subsection{Effect of Scaling on Performance}
To assess the effect of data size on the performance of the algorithms, we had to subject the system to a workload of queries, each with different characteristics. We utilize the TPC-DS \textit{Item} and \textit{Time} workloads listed in Tables \ref{tab:T5} and \ref{tab:T6} and include a set of ANALYZE queries with various combinations of characteristics. 

\silence{The selectivity ratio is set to be 20\% and 50\%, the level of grouper values is set to take one of the values \textit{low-low, low-mid, mid-mid, high-high} 
and the number of join operation varies from 2 to 5. The aggregation 
function is \textit{sum()} and the number of groupers is two.}

Figures \ref{fig:Fa}, \ref{fig:Fb}, \ref{fig:Fc} present the effect of scaling data size on the query execution time across all methods for the \textit{Item} workload. Min-MQO shown as dark blue , Mid-MQO as blue, and Max-MQO as light blue. The number of fact table entries of the same dataset utilized for our evaluation is 2/10/100 million. It is evident that as the data size scales, the query execution time of all algorithms is increasing. Observe that Max-MQO's performance drops significantly as the number of fact table entries increases and the performance gap between the other antagonists is increasing. Min-MQO's performance closely follows that of Mid-MQO, but Mid-MQO is always faster. 
Figures \ref{fig:Fg}, \ref{fig:Fh}, \ref{fig:Fi} present the effect of scaling data size on the query execution time for all methods for the \textit{Time} workload. It is evident that as the data size scales, the query execution time of all algorithms is increasing. In this case, the dimension tables are large, and the sibling facilitator queries do not dominate the execution time of the ANALYZE query. We observe that Max-MQO's performance scales better with the data size, in comparison to its antagonists, and is the fastest algorithm in almost every query. Min-MQO's performance closely follows that of Mid-MQO, but Mid-MQO is always faster.

\begin{figure}[htbp]
    \centering

    \begin{subfigure}[t]{0.3\textwidth}
        \centering
        \includegraphics[width=\linewidth]{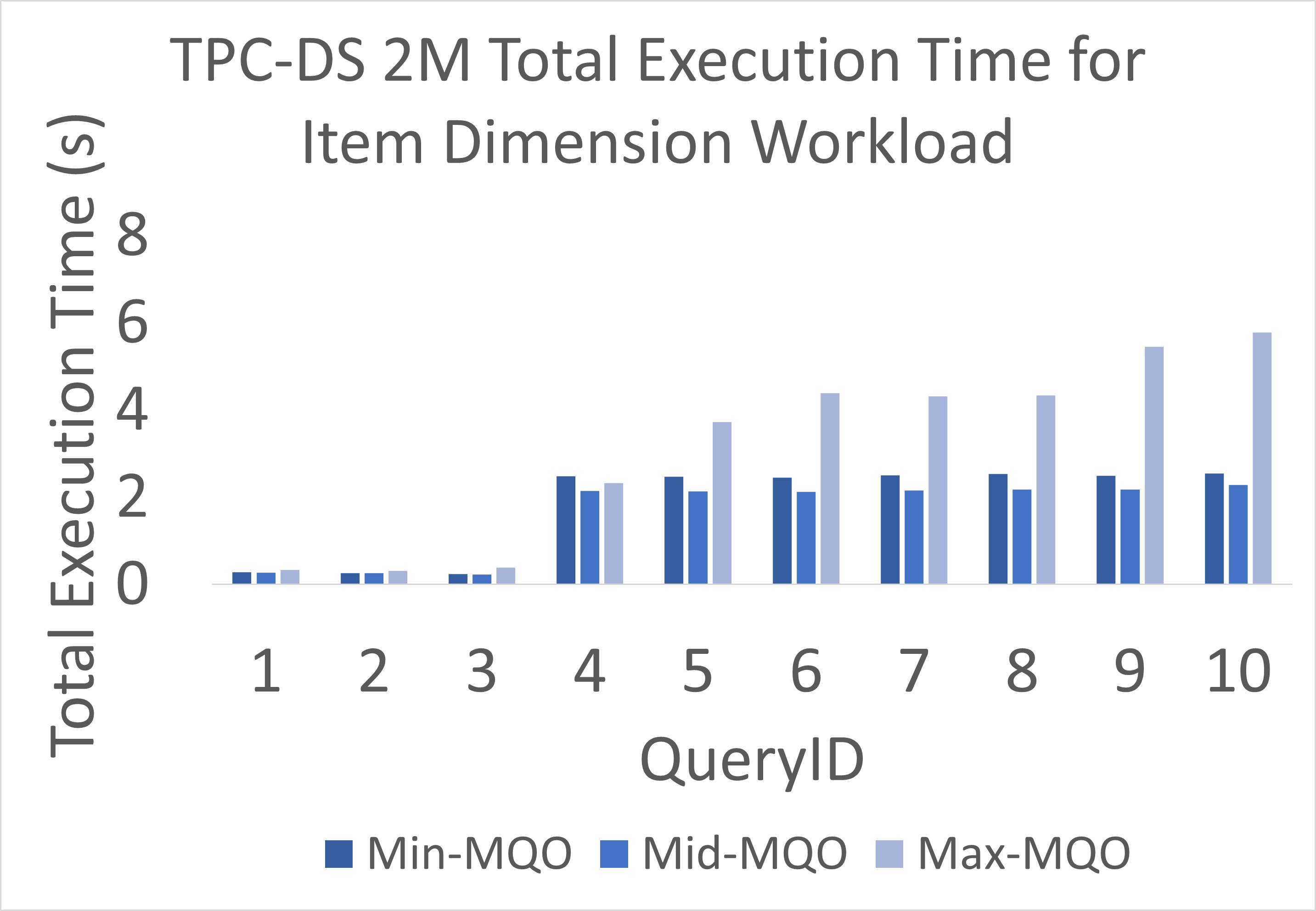}
        \caption{TPC-DS 2M \textit{Item }Workload}
        \label{fig:Fa}
    \end{subfigure}
    \hfill
    \begin{subfigure}[t]{0.3\textwidth}
        \centering
        \includegraphics[width=\linewidth]{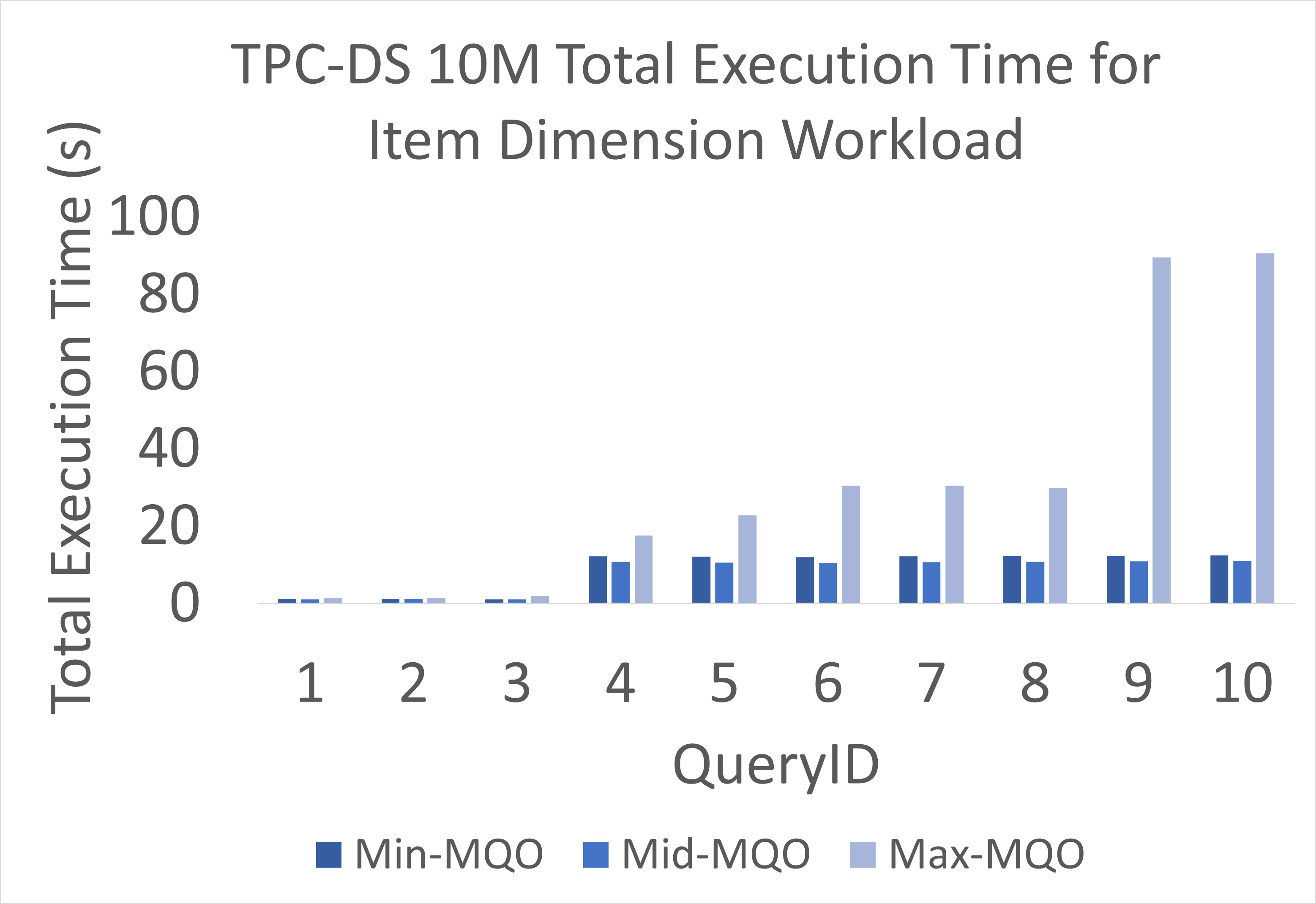}
        \caption{TPC-DS 10M \textit{Item} Workload}
        \label{fig:Fb}
    \end{subfigure}
    \hfill
    \begin{subfigure}[t]{0.3\textwidth}
        \centering
        \includegraphics[width=\linewidth]{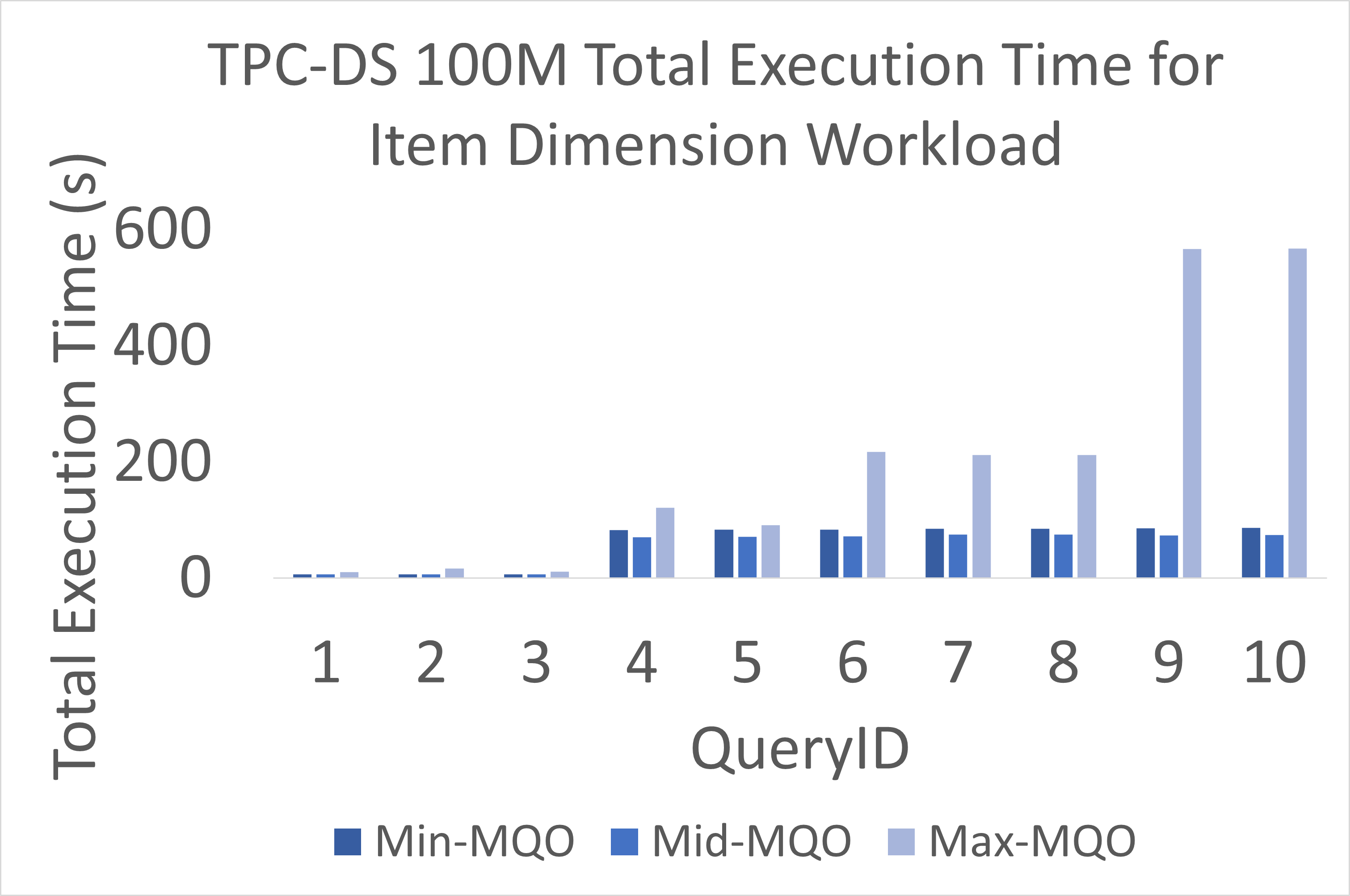}
        \caption{TPC-DS 100M \textit{Item} Workload}
        \label{fig:Fc}
    \end{subfigure}

    \medskip

    \begin{subfigure}[t]{0.3\textwidth}
        \centering
        \includegraphics[width=\linewidth]{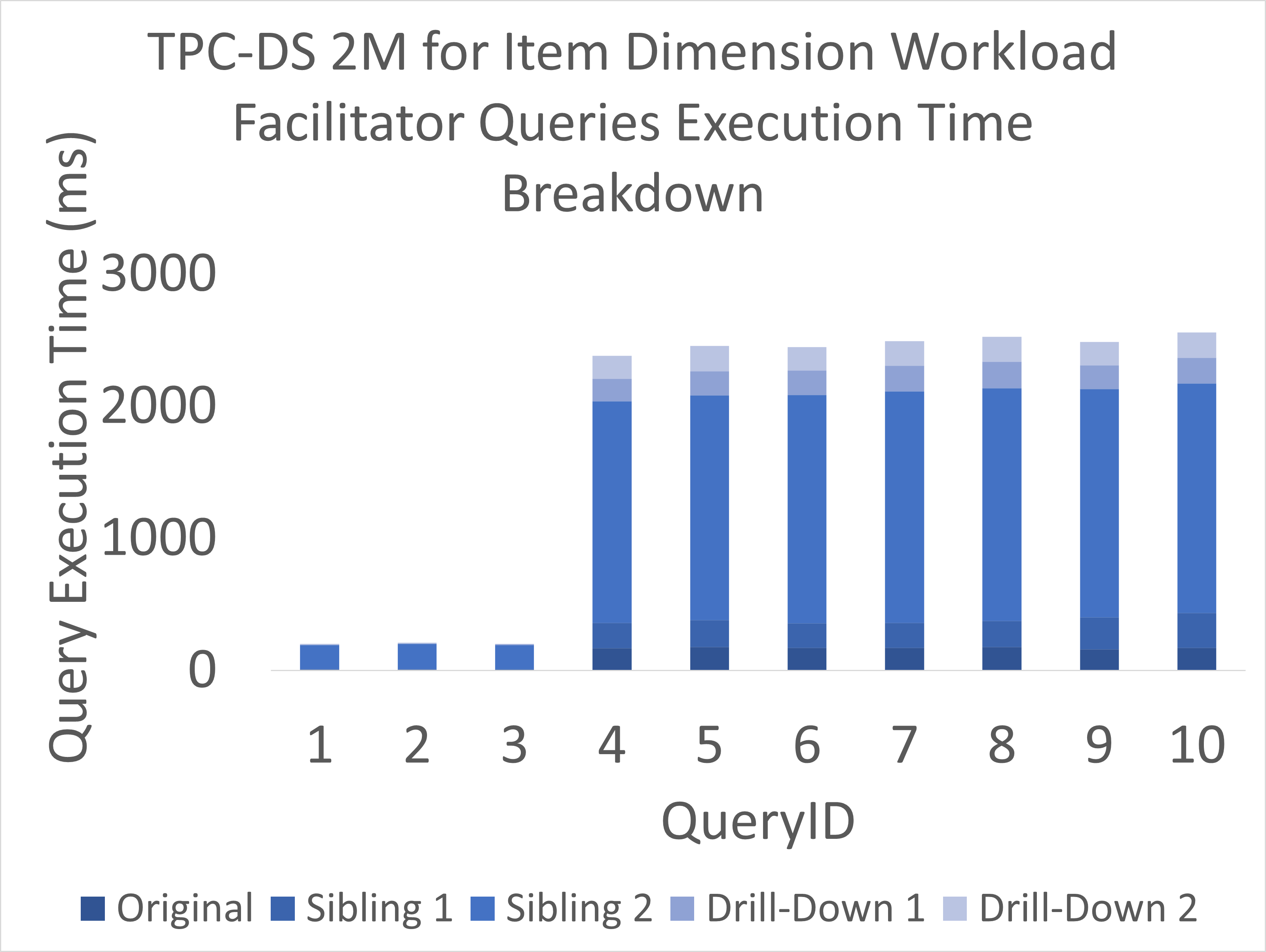}
        \caption{TPC-DS 2M \textit{Item} Workload Facilitator Queries Breakdown}
        \label{fig:Fd}
    \end{subfigure}
    \hfill
    \begin{subfigure}[t]{0.3\textwidth}
        \centering
        \includegraphics[width=\linewidth]{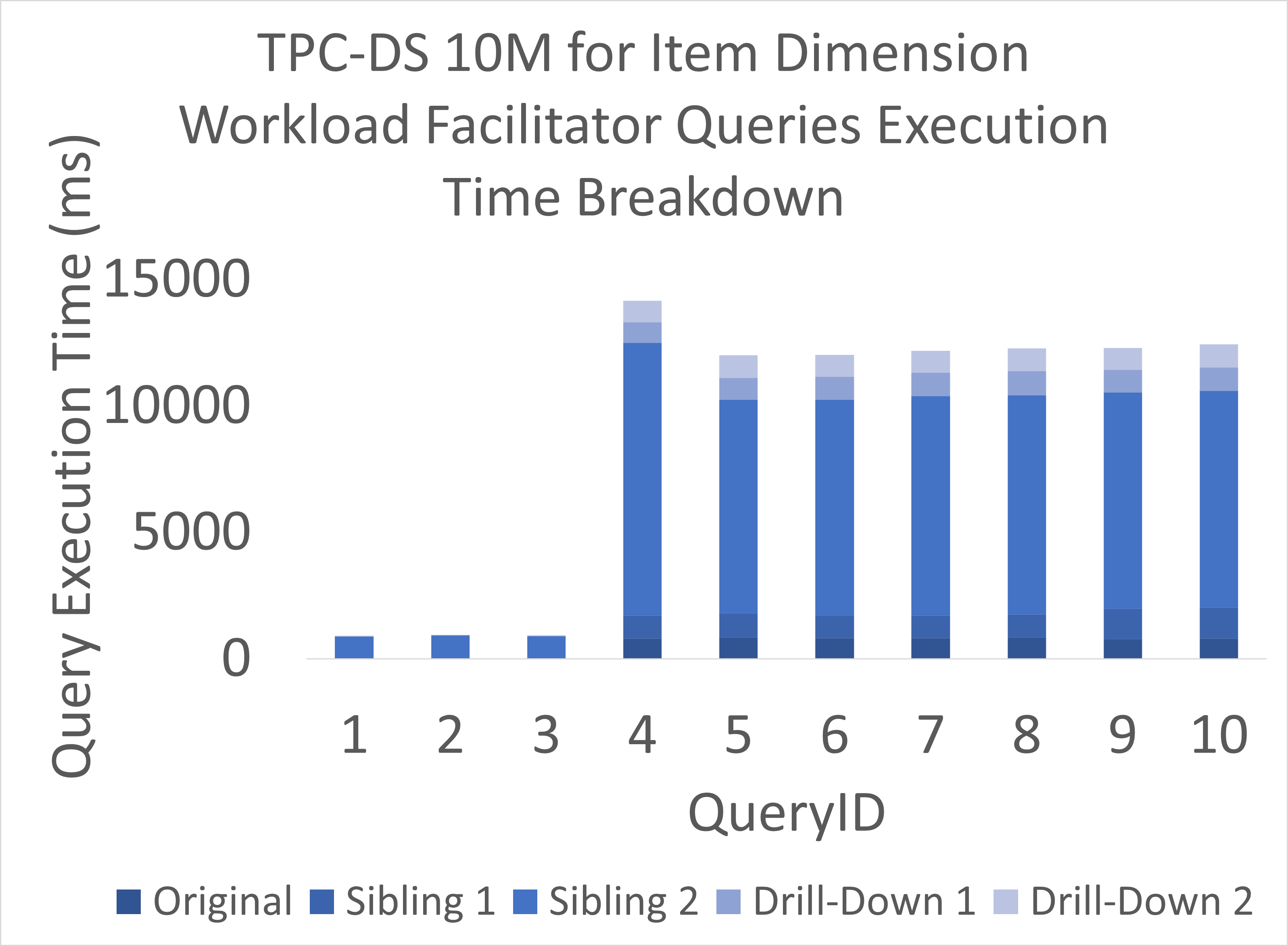}
        \caption{TPC-DS 10M \textit{Item} Workload Facilitator Queries Breakdown}
        \label{fig:Fe}
    \end{subfigure}
    \hfill
    \begin{subfigure}[t]{0.3\textwidth}
        \centering
        \includegraphics[width=\linewidth]{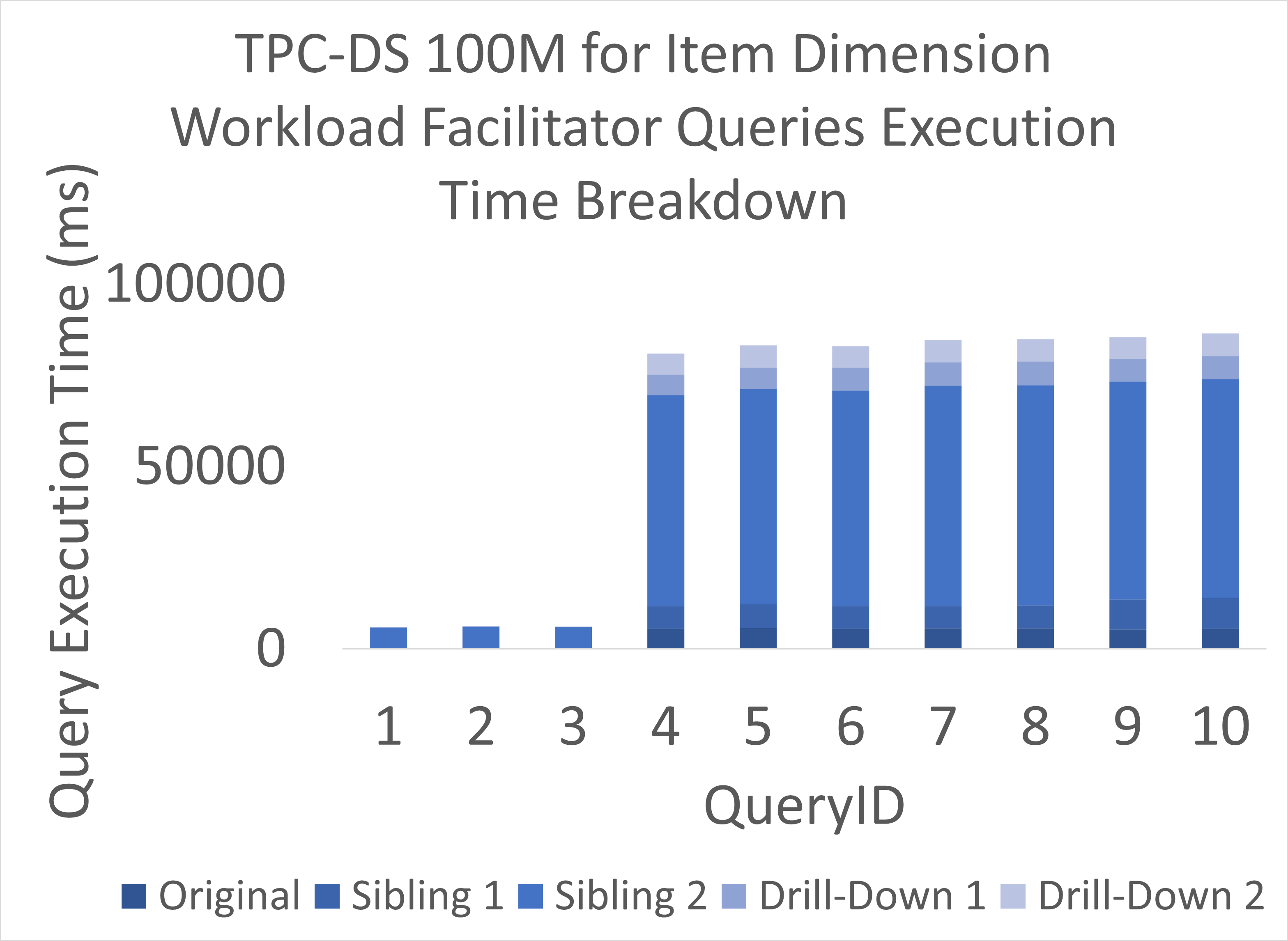}
        \caption{TPC-DS 100M \textit{Item} Workload Facilitator Queries Breakdown}
        \label{fig:Ff}
    \end{subfigure}

\hrule

    \medskip

    \begin{subfigure}[t]{0.3\textwidth}
        \centering
        \includegraphics[width=\linewidth]{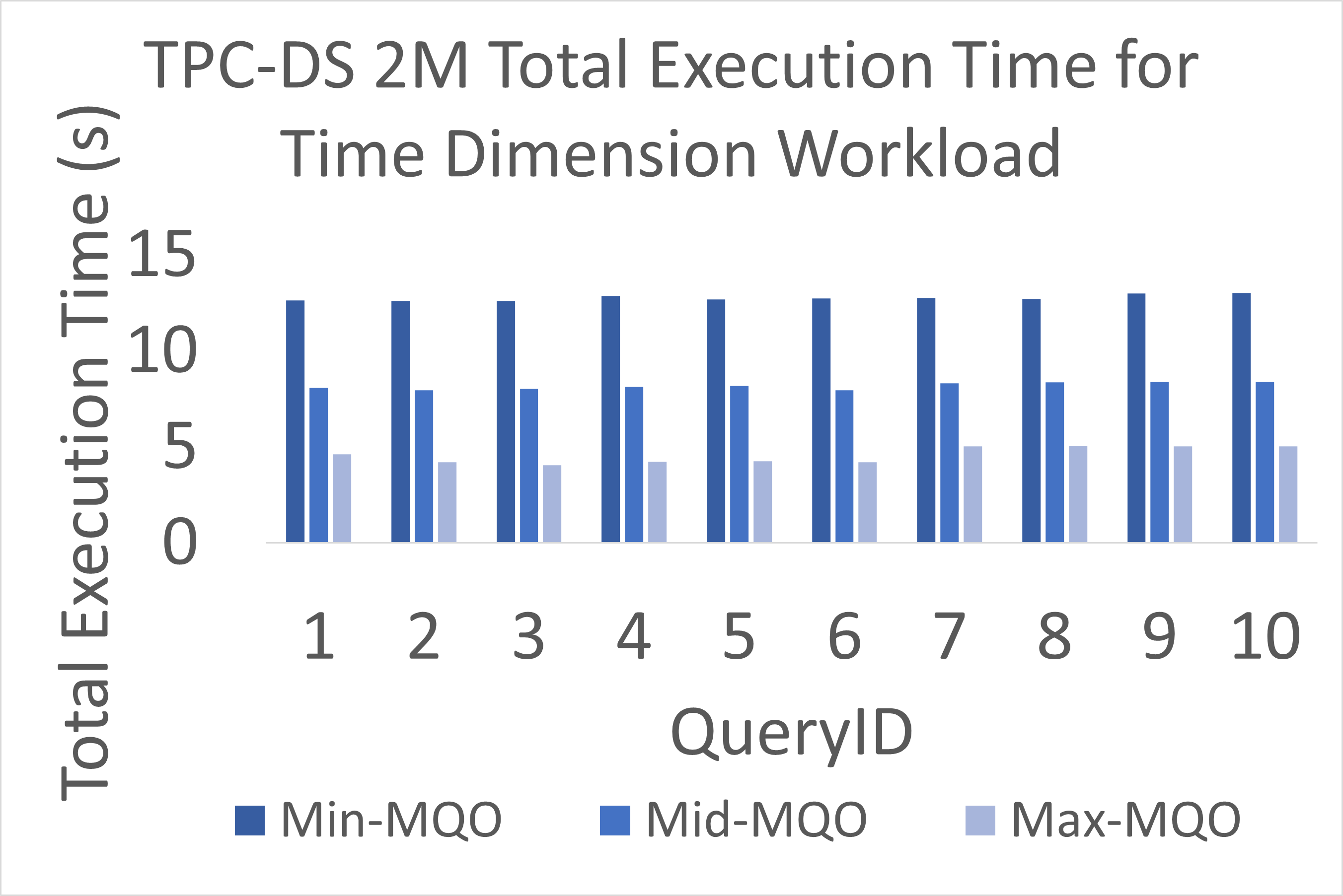}
        \caption{TPC-DS 2M \textit{Time} Workload}
        \label{fig:Fg}
    \end{subfigure}
    \hfill
    \begin{subfigure}[t]{0.3\textwidth}
        \centering
        \includegraphics[width=\linewidth]{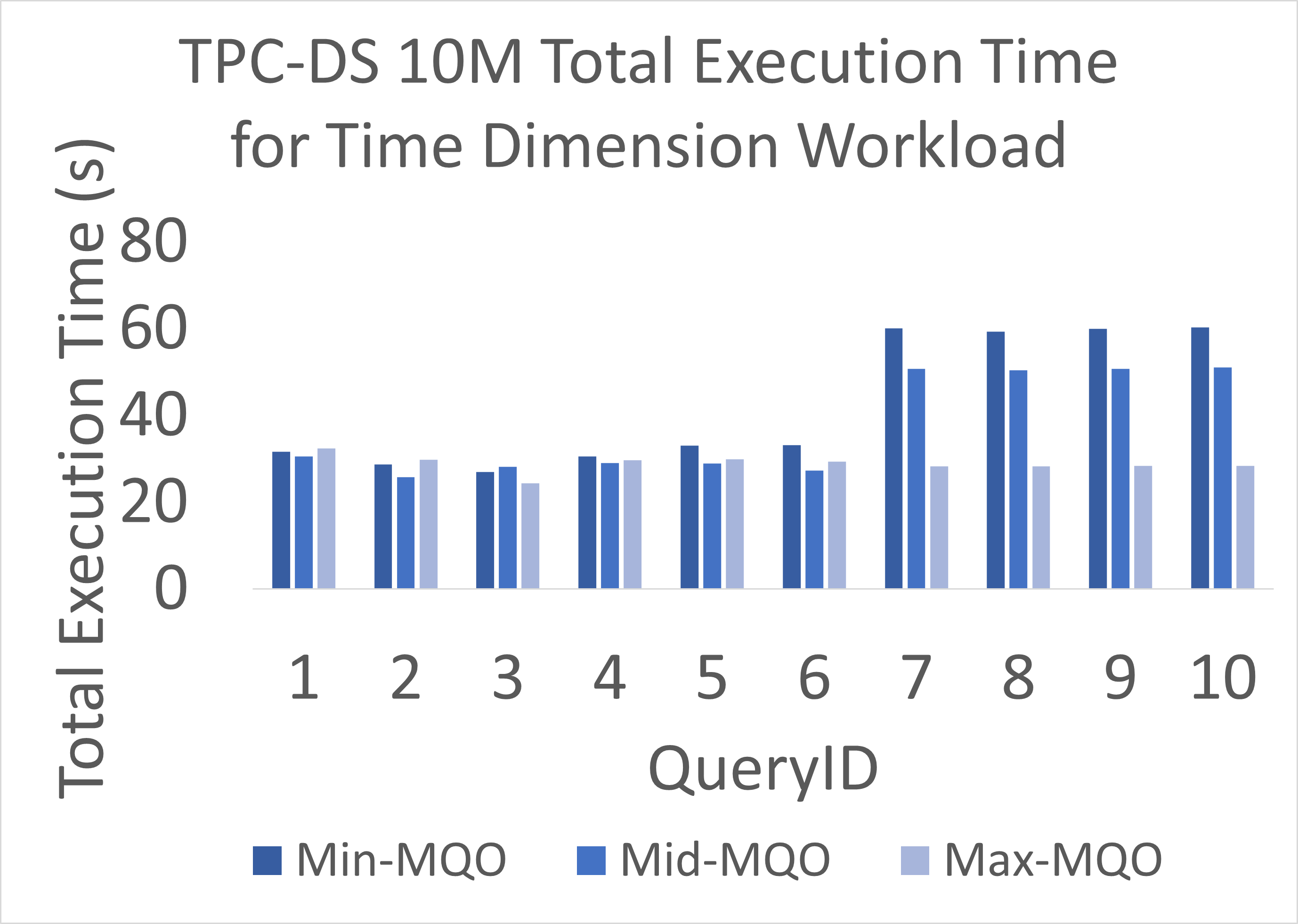}
        \caption{TPC-DS 10M \textit{Time} Workload}
        \label{fig:Fh}
    \end{subfigure}
    \hfill
    \begin{subfigure}[t]{0.3\textwidth}
        \centering
        \includegraphics[width=\linewidth]{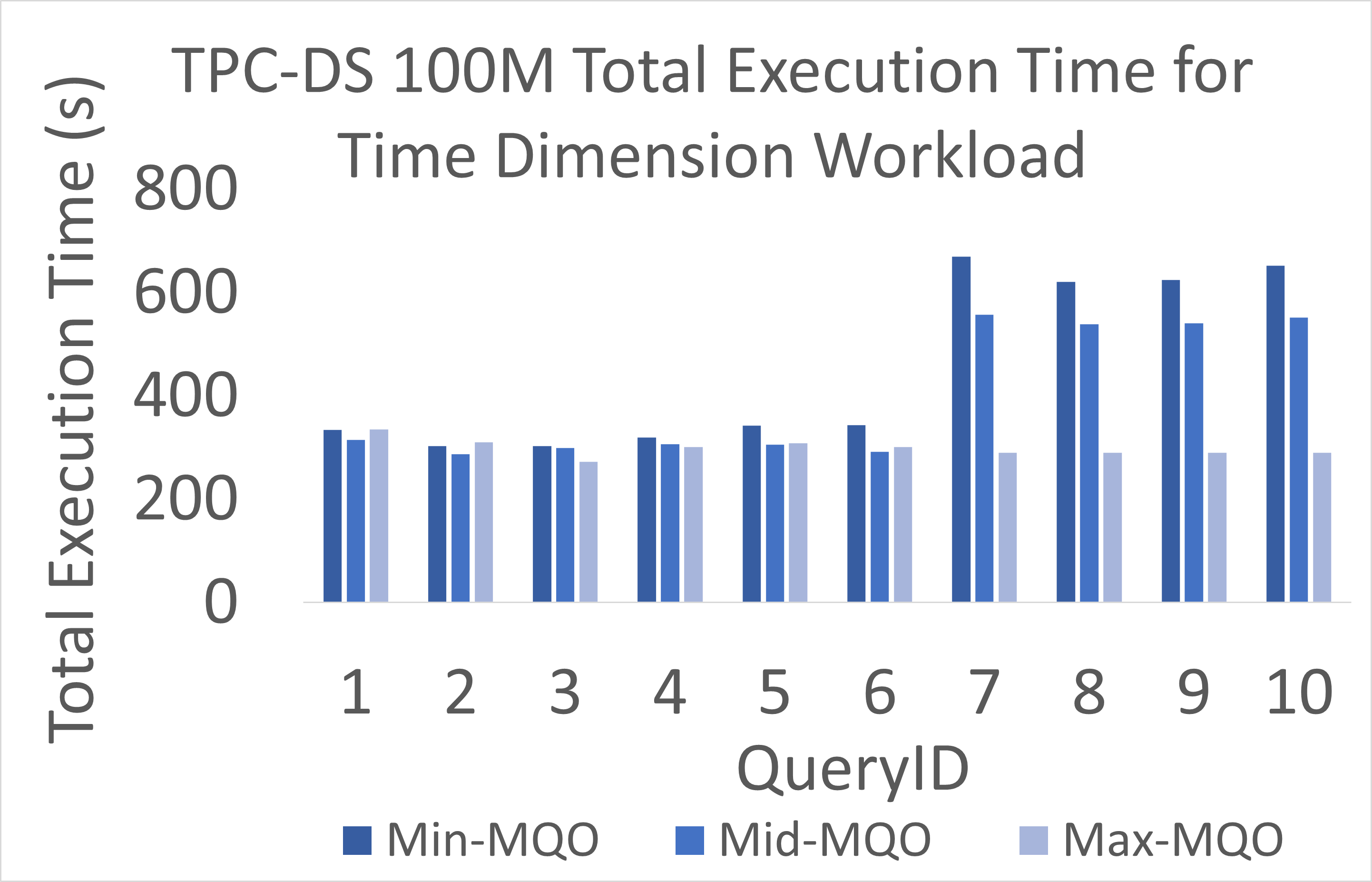}
        \caption{TPC-DS 100M \textit{Time} Workload}
        \label{fig:Fi}
    \end{subfigure}

    \medskip

    \begin{subfigure}[t]{0.3\textwidth}
        \centering
        \includegraphics[width=\linewidth]{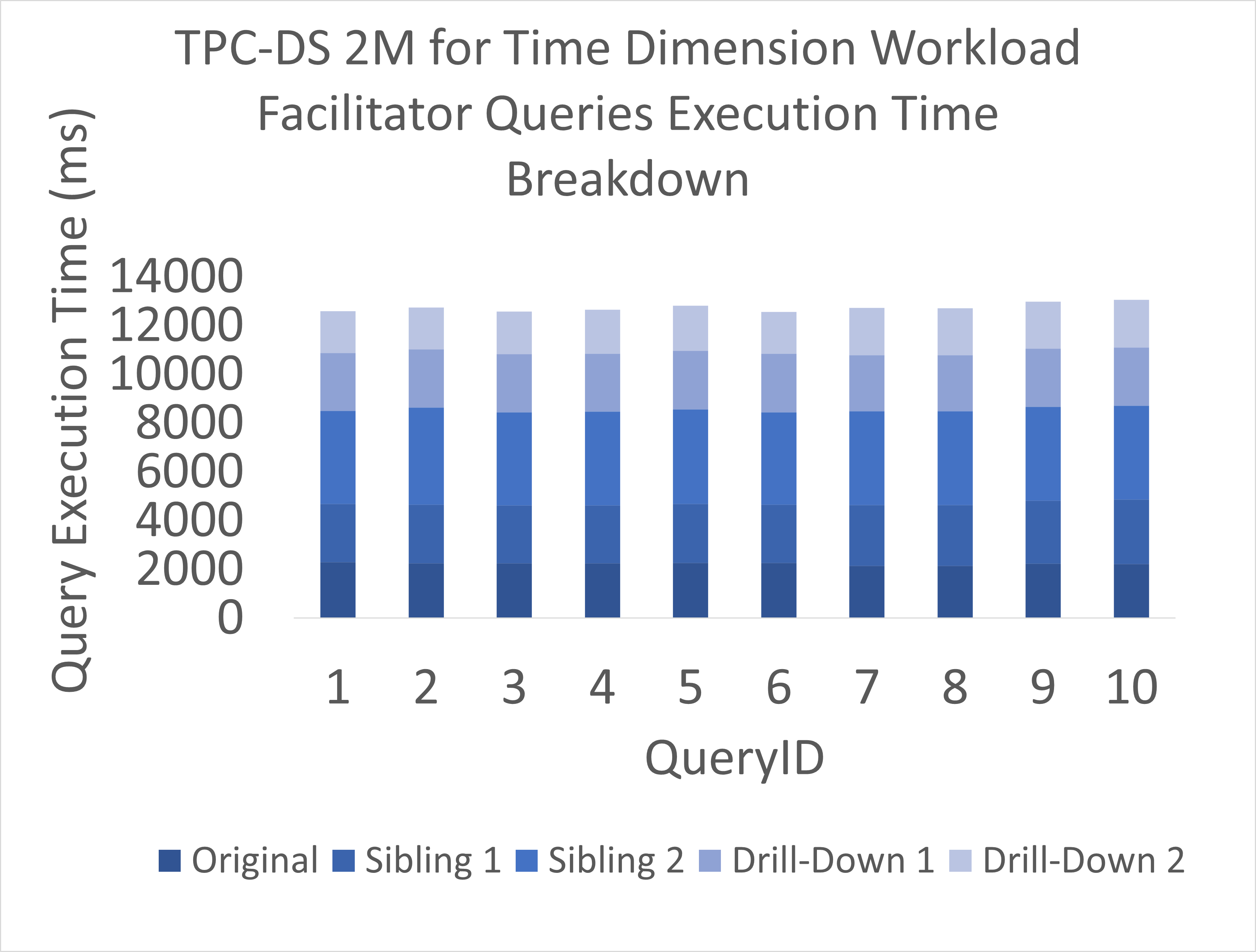}
        \caption{TPC-DS 2M \textit{Time} Workload Facilitator Queries Breakdown}
        \label{fig:Fj}
    \end{subfigure}
    \hfill
    \begin{subfigure}[t]{0.3\textwidth}
        \centering
        \includegraphics[width=\linewidth]{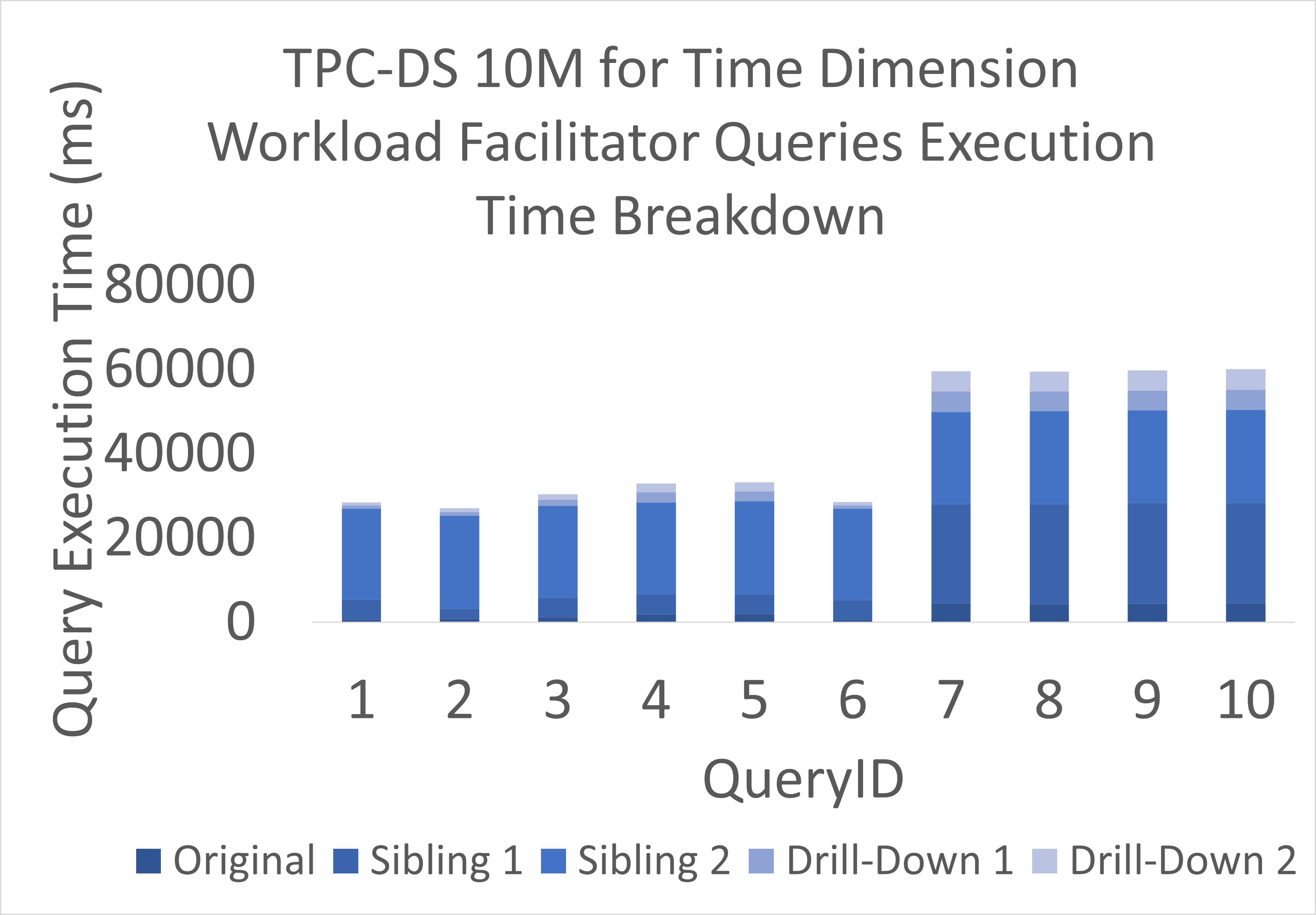}
        \caption{TPC-DS 10M \textit{Time} Workload Facilitator Queries Breakdown}
        \label{fig:Fk}
    \end{subfigure}
    \hfill
    \begin{subfigure}[t]{0.3\textwidth}
        \centering
        \includegraphics[width=\linewidth]{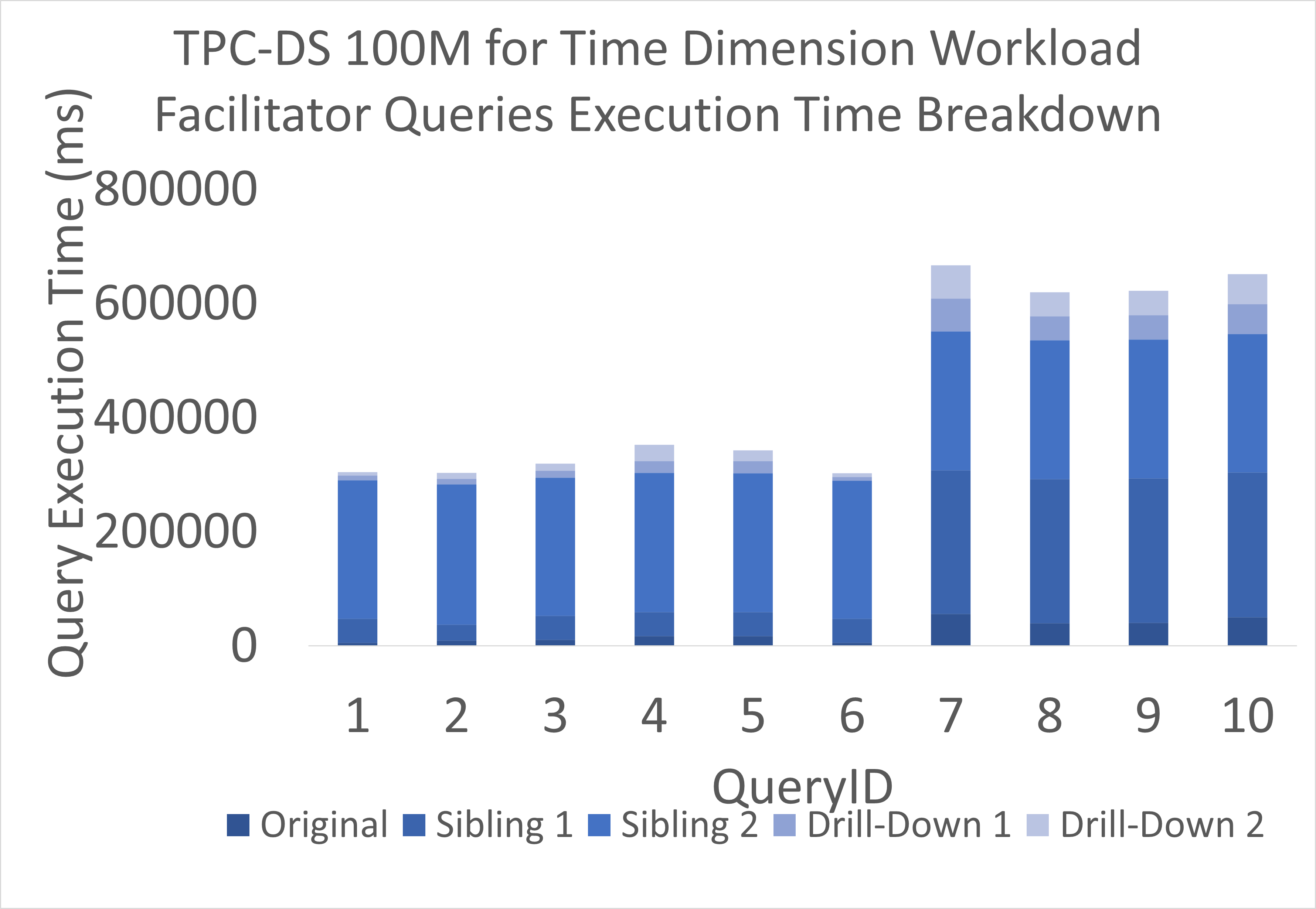}
        \caption{TPC-DS 100M \textit{Time} Workload Facilitator Queries Breakdown}
        \label{fig:Fl}
    \end{subfigure}

\hrule

      \medskip

    \begin{subfigure}[t]{0.3\textwidth}
        \centering
        \includegraphics[width=\linewidth]{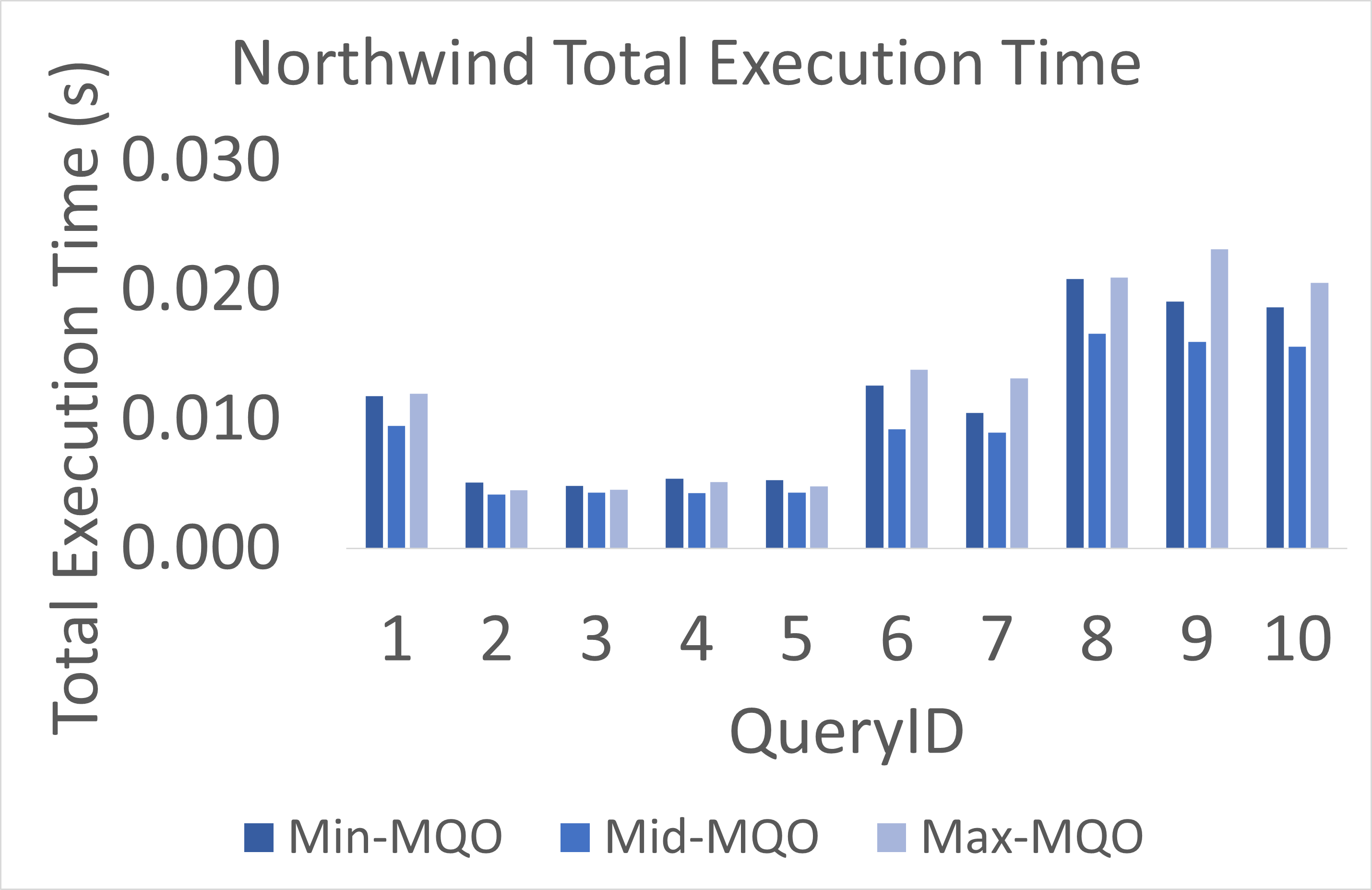}
        \caption{Northwind Workload}
        \label{fig:Fm}
    \end{subfigure}
    \hfill
    \begin{subfigure}[t]{0.3\textwidth}
        \centering
        \includegraphics[width=\linewidth]{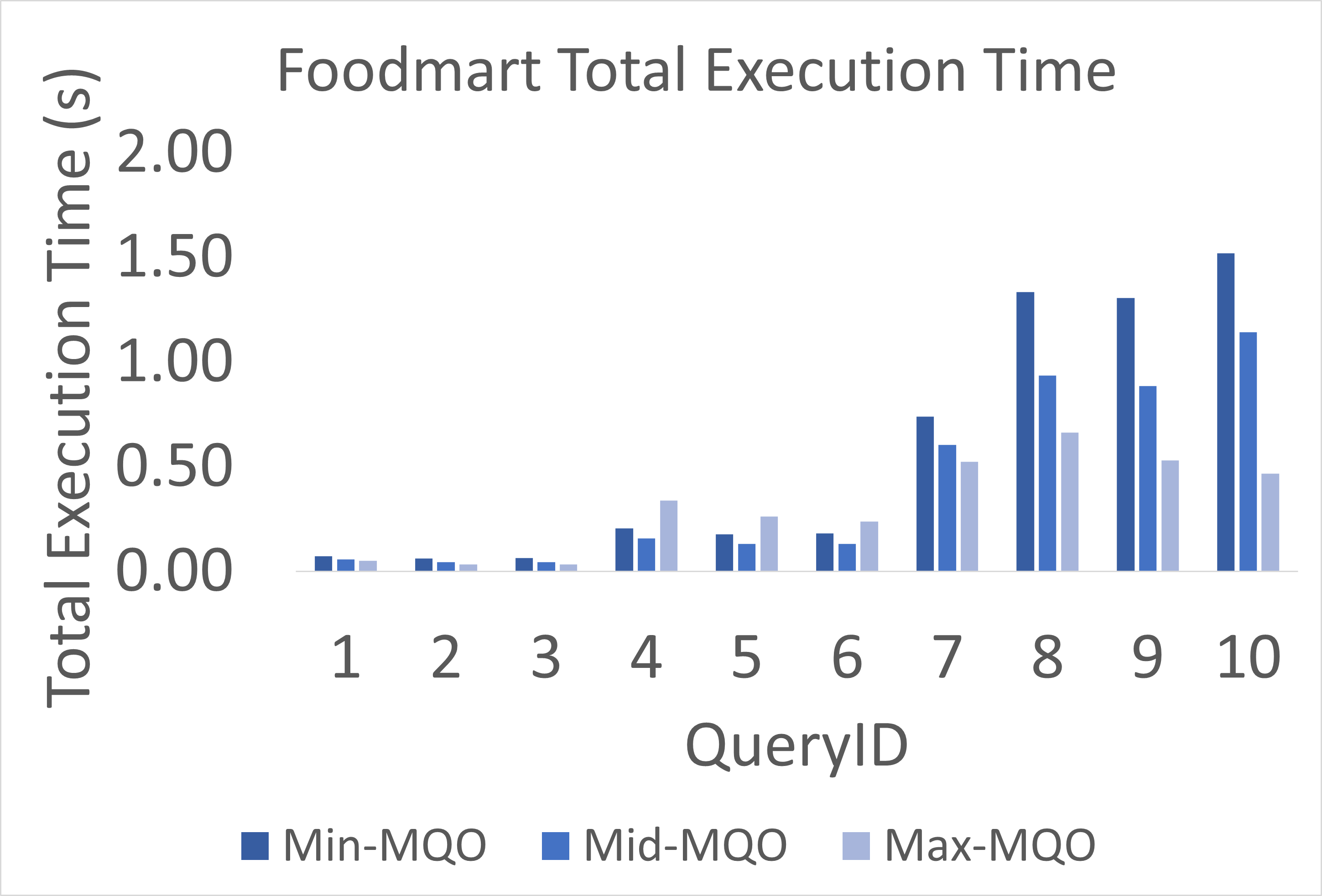}
        \caption{Foodmart Time Workload}
        \label{fig:Fn}
    \end{subfigure}
    \hfill
    \begin{subfigure}[t]{0.3\textwidth}
        \centering
        \includegraphics[width=\linewidth]{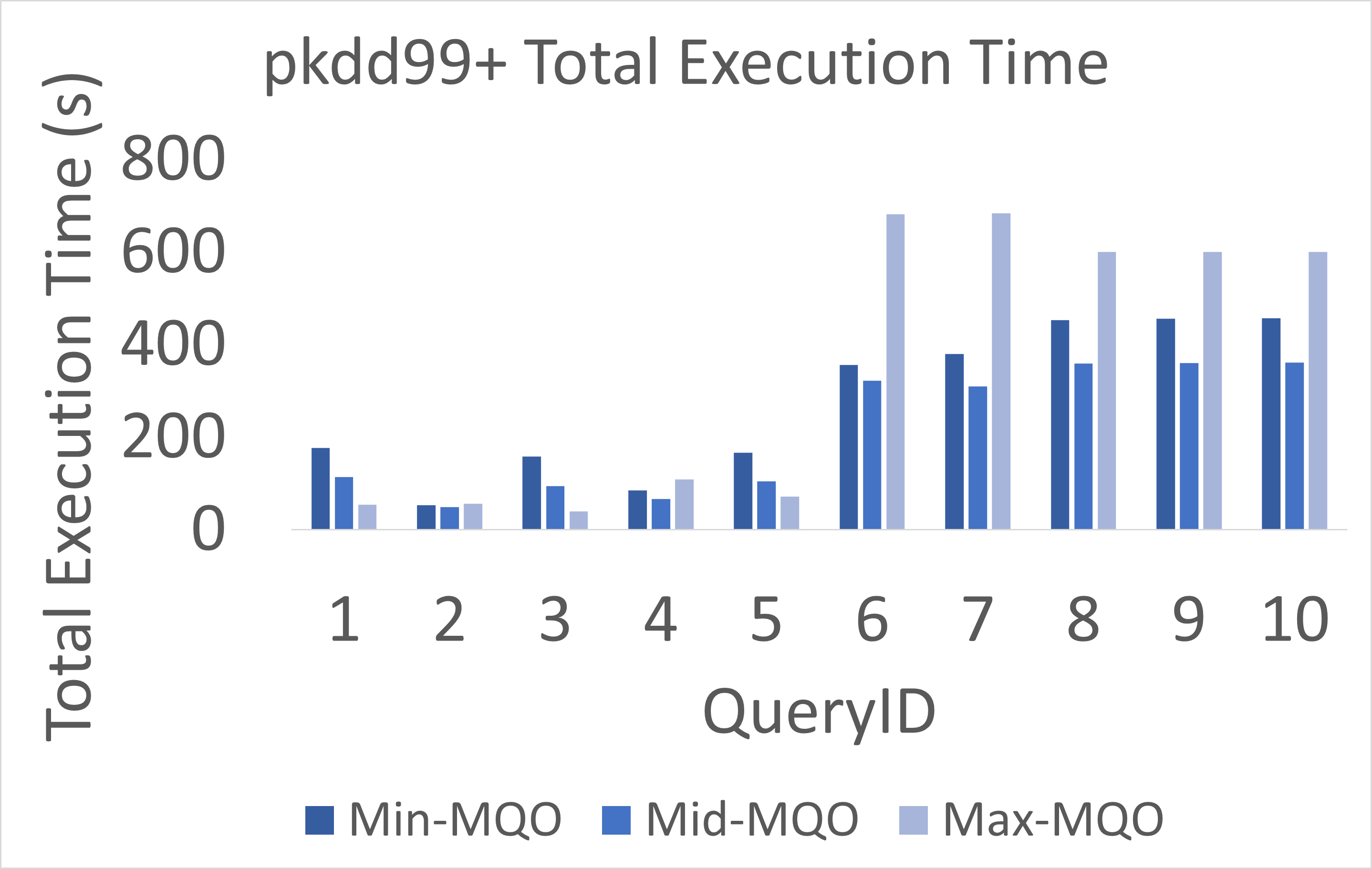}
        \caption{pkdd99+ Time Workload}
        \label{fig:Fo}
    \end{subfigure}

    \caption{Performance Evaluation Figures}
    \label{fig:F}
\end{figure}

\subsection{Performance Evaluation}
To assess the performance of the antagonist methods, we evaluate each algorithm's performance on various datasets and data sizes. Figure \ref{fig:F} presents the performance evaluation for each dataset with respect to the total execution time of each algorithm. For the TPC-DS \textit{Item} and \textit{Time} workloads, Figure \ref{fig:F} also presents the facilitator queries' breakdown for each ANALYZE query of the workload.

\subsubsection{TPC-DS \textit{Item} Workload}
Figures \ref{fig:Fa}, \ref{fig:Fb}, and \ref{fig:Fc} illustrate the \textit{Item} query workload's execution time across all methods for each TPC-DS dataset size variation. Min-MQO is shown as a dark blue line, Mid-MQO as a blue line, Max-MQO as a light blue line. The dimension tables involved in the workload are the \textit{Date} dimension ($\approx$80000 tuples) and the \textit{Item} dimension ($\approx$ 20000 tuples). We observe the following: 
\begin{itemize}
    \item Regardless of the dataset size, Mid-MQO performs the best.
    \item Max-MQO constantly performs the worst, especially when the selectivity ratio is increasing.
    \item Min-MQO and Mid-MQO do not differ significantly in execution time.
\end{itemize}

\silence{(i) Regardless of the dataset size, Mid-MQO performs the best, (ii) Max-MQO constantly performs the worst, especially when the selectivity ratio is increasing, and, (iii) Min-MQO and Mid-MQO do not differ significantly in their execution time.} 

Figures \ref{fig:Fd},\ref{fig:Fe},\ref{fig:Ff} provide the detailed execution time for each facilitator query of the workload's ANALYZE queries. We observe that the drill-down queries that Mid-MQO does not execute in comparison to Min-MQO, are not very significant for the total execution time. Thus, the difference in performance in Min-MQO and Mid-MQO is relatively small. Max-MQO processes a large number of tuples compared to the other antagonists and, in that context, underperforms constantly. 

\subsubsection{TPC-DS \textit{Time} Workload}
Figures \ref{fig:Fg},\ref{fig:Fh},\ref{fig:Fi} illustrate the \textit{Time} query workload's execution time across all methods for each TPC-DS dataset size variation. The dimension tables involved in the workload are the \textit{Date} dimension ($\approx$80000 tuples) and the \textit{Time} dimension ($\approx$ 70000 tuples). We observe the following. 
\begin{itemize}
    \item In the TPC-DS 2M, Max-MQO outperforms the other algorithms, and Min-MQO is the worst performing algorithm in every case.
    \item In the TPC-DS 10M and 100M, in small selectivity ratios, Max-MQO performance is almost stable regarding the selectivity ratio increase, while Min-MQO's and Mid-MQO's performance deteriorates as the selectivity ratio increases.
\end{itemize}

\silence{(i) In the TPC-DS 2M, Max-MQO outperforms the other algorithms and Min-MQO is the worst performing algorithm in every case, (ii) in the TPC-DS 10M and 100M, in small selectivity ratios, Max-MQO performance is almost stable regarding the selectivity ratio increase, while Min-MQO's and Mid-MQO's performance deteriorates as the selectivity ratio increases.}

Here, due to the large size of both dimension tables, the performance of Max-MQO is stable, as it consistently explores a large subset of the fact table once, and avoids the heavy extra cost of siblings of the other algorithms when the selectivity increases (see Figures \ref{fig:Fj}, \ref{fig:Fk}, \ref{fig:Fl} for the breakdown, especially for high-selectivity queries with large QueryID's). The breakdown of the execution time in facilitator queries is presented in Figure \ref{fig:Fj}, \ref{fig:Fk}, \ref{fig:Fl} showing that as the selectivity ratio increases, the execution cost of the drill-down queries justifies the increase in the total execution time difference between Min-MQO and Mid-MQO, thus increasing the performance difference between Min-MQO and Mid-MQO as well.

\subsubsection{Other datasets}
Finally, when comparing the algorithms on the 3 other datasets (Figure \ref{fig:Fm}, \ref{fig:Fn}, \ref{fig:Fo}) we observe.
\begin{itemize}
    \item In the Northwind workload (Table \ref{tab:T9}) that involves two small dimension tables ($\approx$90 tuples and $\approx$9 tuples), Mid-MQO constantly scores wins against the other antagonists, while Max-MQO is underperforming.
    \item In the Foodmart workload (Table \ref{tab:T8}) that involves unbalanced dimension tables ($\approx$10000 tuples and $\approx$1000 tuples), Max-MQO is the fastest algorithm, while Min-MQO is the slowest.
    \item In the pkdd99+ workload (Table \ref{tab:T7}) that involves one fairly large and one small dimension ($\approx$30000 tuples and $\approx$5000), Mid-MQO outperforms the other antagonists in the high selectivity ratios, while Max-MQO performs better in the low selectivity ratios.
\end{itemize}

\begin{figure}
  \begin{center}
    \includegraphics[width=0.9\textwidth]{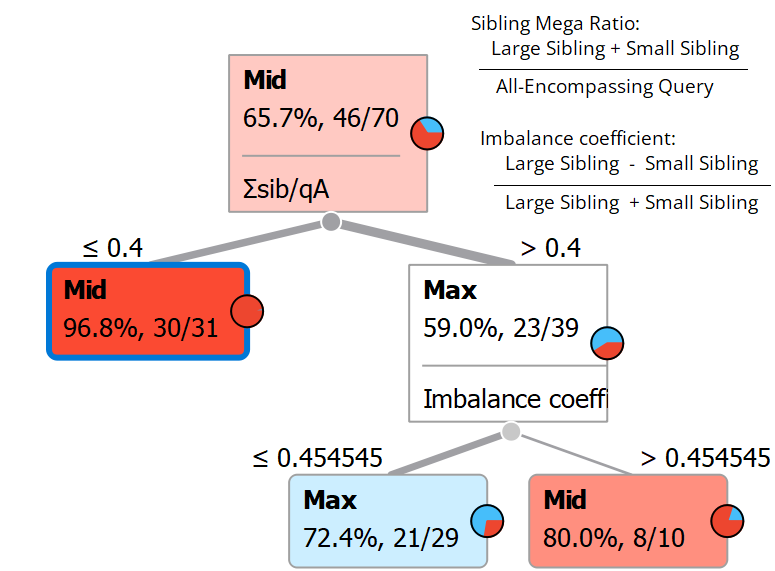}
  \end{center}
  \caption{Optimal Algorithm Criteria}
      \label{fig:DecisionTree}
\end{figure}

\subsection{Effect of Facilitator Queries Performance}
Figures \ref{fig:Fd},\ref{fig:Fe},\ref{fig:Ff} present the execution time for all five facilitator queries for the TPC-DS \textit{Item} workload. We observe that Min-MQO's performance closely follows Mid-MQO's performance regardless of dataset size (Figures \ref{fig:Fa},\ref{fig:Fb},\ref{fig:Fc}). The drill-down queries' execution time for the TPC-DS Item workload is a small percentage of the total execution time of the facilitator queries. However, the execution time of the siblings is a very large percentage of the total execution time. Since both algorithms execute the sibling queries and Mid-MQO omits the drill-down's execution, it is evident that the performance of Min-MQO will follow the performance of Mid-MQO.

Figures \ref{fig:Fj}, \ref{fig:Fk}, \ref{fig:Fl} illustrate the execution time for all five facilitator queries for the TPC-DS \textit{Time} workload. In the case of 2 million facts, the performance gap between Min-MQO and Mid-MQO is significant. We observe that the drill-down queries' execution time is almost 1/3 of the total execution time. In this context, we can state that Mid-MQO which does not execute the drill-down queries, gains a performance advantage over Min-MQO. In all other data sizes (10 million and 100 million), the execution time of the drill-down's is not as large as in the 2 million case and the performance gap between Min-MQO and Mid-MQO is shrinking. 

Taking into account the aforementioned observations, we can state that the execution time of the drill-down queries can affect the performance gap between Min-MQO and Mid-MQO. If the execution time of the drill-down queries is a significant percentage of the total execution time, then Mid-MQO greatly outperforms Min-MQO. However, if the execution time of the drill-down queries is not as significant, then Min-MQO is outperformed by Mid-MQO but follows Mid-MQO's performance closely. 

When does this happen? The crux of the distinction does not lie at the drill down queries, but rather at the extent of the sibling queries. The space of detailed fact tuples that the drill down queries explore is exactly the same with the one of the original query; it is only the eventual aggregation levels that differ. On the other hand, depending on the siblings involved, sibling queries can span a small subset of the detailed facts, or, a very large one. Thus, their ratio over the drill down queries changes too.

\subsection{Choosing Algorithms} 
\result{Mid-MQO is a safe choice as a query processing strategy, as it is both stable and winning most of the time}. However, can we do better than this first result? We have analyzed the entire corpus of 70 workloads (Tables \ref{tab:T5}, \ref{tab:T6}, \ref{tab:T7}) of the large TPC-DS and pkdd99+ datasets (7 workloads of 10 queries each) to \textit{precisely detect }when each query processing strategy wins, via easily computable cost measures. Ultimately (Figure~\ref{fig:DecisionTree}), \result{Max-MQO wins when (a) the siblings touch a fairly large number of fact tuples compared to the all-encompassing query (more than 40\%), and, (b) there is a fairly small imbalance between the sibling queries (less than 45\%). In all other cases, Mid-MQO wins.} The rationale is simple:

\begin{enumerate}
    \item when $q^A$ is too large compared to facilitator queries, then, it explores too much of the data space, hence Max-MQO loses.
    \item when the siblings are both sizable and imbalanced, they do not significantly overlap, thus Max-MQO, visiting them once, does not save time.
\end{enumerate}

The only case where Max-MQO saves time is if sizable siblings overlap (estimated via the imbalance coefficient). Naturally, the exact thresholds are an engineering default, subject to a wider evaluation over time; however, the rationale is consistent with the theoretical setup of $q^A$.

\section{Conclusions}\label{sec:end}
The paper introduces ANALYZE, a novel intentional cube query operator that provides a 360° view of the data by combining original, sibling, and drill-down queries in a single invocation with formal semantics. We have shown that the operator internal queries can be safely merged with theoretical correctness guaranties, enabling principled multi-query optimization at the operator level without modifying the underlying query optimizer. We have introduced, implemented in a data analytics system, and extensively evaluated three execution strategies, demonstrating that partial merging (Mid-MQO) delivers robust and scalable performance across workloads, and full merging (Max-MQO) is beneficial only when the siblings are sizable and significantly overlap, while the not-optimized strategy (Min-MQO) is consistently suboptimal.

As part of future work, several paths can be followed. 
Right now, a plain DBMS would simply execute Min-MQO to process the set of facilitator queries of the operator. A lesson learned here is that the external data analytics engine \textit{can} optimize query execution \textit{outside the DBMS}. This is a lesson to be applied to future intentional operators as well (PREDICT, EXPLAIN, etc). 
Exploiting the formal, algebraic nature of the operator also raises the question of integrating it with cost-based DBMS optimizers or platforms like Apache Spark. 
Several possible extensions concern the internal design choices of the current version of ANALYZE. Relaxing some of the constraints (e.g., by more groupers, siblings based on user-defined or context-dependent benchmarks \cite{DBLP:journals/ijpcc/StefanidisPV07, DBLP:journals/air/MateosB25}) is also open to research.
For example, generalizing our results to more than two groupers is possible; we conjecture that despite the formalization overhead that this incurs, there is no limit to the number of groupers that can be involved in the formulation of an ANALYZE query. As another example, regulating or even automating the benchmark against which the original query is contrasted is also possible (e.g., comparing sales of a quarter to the same quarter of a previous year), with the respective impact to the query processing algorithms. 


\paragraph{Acknowledgments.}
The research project is implemented in the framework of
H.F.R.I call “3rd Call for H.F.R.I.’s Research Projects to Support Faculty Members
\& Researchers” (H.F.R.I. Project Number: 23640). 

D. Gkitsakis helped with early, trial versions of the code.


\section{Artifacts}
The source code, data and/or other artifacts have been made available at \url{https://github.com/DAINTINESS-Group/DelianCubeEngine}. Navigate to the package \textsf{analyzemqoexperiments} at 
\url{https://tinyurl.com/bddf4tzv}
where there is a markdown file with directions on how to repeat the experiments.

\bibliographystyle{alpha}
\bibliography{bibliography}


\end{document}